\documentclass[a4paper,11pt]{article}

\usepackage[margin=2cm]{geometry}
\usepackage[utf8]{inputenc}

\usepackage[font=footnotesize,labelfont=bf]{caption}
\usepackage[font=scriptsize,labelfont=bf]{subcaption}

\usepackage{amsthm}
\usepackage{amsmath}
\usepackage{amssymb}
\usepackage{commath}
\usepackage{mathtools}
\usepackage{forest}
\usepackage{xcolor}
\usepackage{url}
\usepackage{nicefrac}

\newtheorem{theorem}{Theorem}
\newtheorem{lemma}{Lemma}

\usepackage{csquotes}
\usepackage[ruled,linesnumbered]{algorithm2e}
\usepackage{cleveref}
\usepackage{xspace}
\usepackage{tikz}
\usetikzlibrary{patterns, plotmarks, decorations.pathreplacing, trees, shapes}

\renewcommand\and{\end{tabular}\kern-\tabcolsep\ and\ \kern-\tabcolsep\begin{tabular}[t]{c}}
\let\origthanks\thanks
\renewcommand\thanks[1]{\begingroup\let\rlap\relax\origthanks{#1}\endgroup}

\let\originalleft\left
\let\originalright\right
\renewcommand{\left}{\mathopen{}\mathclose\bgroup\originalleft}
\renewcommand{\right}{\aftergroup\egroup\originalright}

\DeclareMathOperator*{\argmin}{arg\,min}

\newcommand{\KofN}{$k$-of-$n$}
\newcommand{\GGHK}{Gkenosis~et~al.}

\newcommand{\ALG}{\ensuremath{\mathrm{A}}\xspace}
\newcommand{\OPT}{\ensuremath{\mathrm{OPT}}\xspace}
\newcommand{\Fail}{\ensuremath{\mathrm{fail}}\xspace}%
\newcommand{\Succ}{\ensuremath{\mathrm{succ}}\xspace}%
\newcommand{\ALGSz}{\ensuremath{\ALG_{\Fail}}\xspace}
\newcommand{\ALGSo}{\ensuremath{\ALG_{\Succ}}\xspace}
\newcommand{\Ch}{\ensuremath{\mathrm{ch}}\xspace}%
\newcommand{\ALGCh}{\ensuremath{\ALG_{\Ch}}\xspace}
\newcommand{\RdRn}{\ensuremath{\textsc{RoundRobin}}\xspace}
\newcommand{\TwRR}{\ensuremath{\mathrm{3R}}\xspace}
\newcommand{\TwoRR}{\ensuremath{\mathrm{2R}}\xspace}
\newcommand{\GenRR}{\ensuremath{\mathrm{GR}}\xspace}
\newcommand{\Inst}{\ensuremath{I}\xspace}
\newcommand{\PhaseTwoCommon}{\ensuremath{E}\xspace}
\newcommand{\PhaseTwoActive}[1]{\ensuremath{\PhaseTwoCommon_{#1}}\xspace}
\newcommand{\PhaseTwoInactive}[1]{\ensuremath{\overline{\PhaseTwoCommon}_{#1}}}
\newcommand{\Hsp}{\ensuremath{H}\xspace}
\newcommand{\HSfail}{\Fail\xspace}
\newcommand{\HSsucc}{\Succ\xspace}
\newcommand{\ChSet}{\ensuremath{C}\xspace}
\newcommand{\UnqSet}{\ensuremath{U}\xspace}
\newcommand{\PrivSet}{\ensuremath{P}\xspace}
\newcommand{\VizSet}{\ensuremath{S}\xspace}

\newcommand{\RB}[1]{\ensuremath{\left( {#1} \right)}}
\newcommand{\SB}[1]{\ensuremath{\left[ {#1} \right]}}
\newcommand{\CB}[1]{\ensuremath{\{ {#1} \}}}

\newcommand{\Ex}[1]{\ensuremath{\mathbb{E} \SB{#1}}}
\newcommand{\Exc}[2]{\Ex{\left. #1 \;\middle|\; #2 \right.}}
\newcommand{\Pty}[1]{\ensuremath{\mathrm{Pr}\SB{{#1}}}}
\newcommand{\Ptyc}[2]{\ensuremath{\Pty{\left. {#1} \;\middle|\; {#2} \right.}}}

\newcommand{\Cost}[3]{\ensuremath{\cost_{#1}\RB{{#2}, {#3}}}}
\newcommand{\cost}{\ensuremath{c}}

\newcommand{\ALPHSo}{\ensuremath{\alpha_{\mathrm{succ}}}}
\newcommand{\ALPHSz}{\ensuremath{\alpha_{\mathrm{fail}}}}
\newcommand{\ALPHCh}{\ensuremath{\alpha_{\Ch}}}

\newcommand{\Sig}[1]{\ensuremath{\sigma_{#1}}}%
\newcommand{\SigVal}[2]{\Sig{#1}\RB{#2}}

\title{Simple Algorithms for Stochastic Score Classification\\ with Small Approximation Ratios}

\author{Benedikt M.\ Plank\thanks{Berlin, Germany. Email: \texttt{b.plank@mail.de}. Supported in part by a research scholarship as well as the TopMath program of the Technical University of Munich, both part of the Elite Network of Bavaria, as well as by the Alexander von Humboldt Foundation with funds from the German Federal Ministry of Education and Research (BMBF).} \and Kevin Schewior\thanks{Department of Mathematics and Computer Science, University of Southern Denmark, Odense, Denmark. Email: \texttt{kevs@sdu.dk}. Supported in part by the Independent Research Fund Denmark, Natural Sciences, grant DFF-0135-00018B.}}

\begin{document}

\maketitle

\begin{abstract}
We revisit the Stochastic Score Classification (SSC) problem introduced by Gkenosis et al.\ (ESA~2018):
We are given $n$ tests. Each test $j$ can be conducted at cost $c_j$, and it succeeds independently with probability $p_j$.
Further, a partition of the (integer) interval $\{0,\dots,n\}$ into $B$ smaller intervals is known. The goal is to conduct tests so as to determine that interval from the partition in which the number of successful tests lies while minimizing the expected cost. Ghuge et al.\ (IPCO 2022) recently showed that a polynomial-time constant-factor approximation algorithm exists.

We show that interweaving the two strategies that order tests increasingly by their $c_j/p_j$ and $c_j/(1-p_j)$ ratios, respectively---as already proposed by Gkensosis et al.\ for a special case---yields a small approximation ratio. We also show that the approximation ratio can be slightly decreased from $6$ to $3+2\sqrt{2}\approx 5.828$ by adding in a third strategy that simply orders tests increasingly by their costs. The similar analyses for both algorithms are nontrivial but arguably clean. Finally, we complement the implied upper bound of $3+2\sqrt{2}$ on the adaptivity gap with a lower bound of $3/2$. Since the lower-bound instance is a so-called unit-cost \KofN{} instance, we settle the adaptivity gap in this case.

\end{abstract}

\section{Introduction}

The \emph{Stochastic Score Classification} (SSC) problem was introduced by Gkenosis et al.~\cite{GkenosisGHK:18} to model situations in which tests can be made in order to evaluate a certain score, e.g., the risk score of a patient. Here, the set of tests is $N:=\{1,\dots,n\}$. Any test $j\in N$ can be conducted at most once at known cost $c_j$, and it \emph{succeeds} independently with known probability $p_j\in(0,1)$ (otherwise it \emph{fails}). The different score classes are given through a partition of the integer interval $\{0,1,\dots,n\}$ into smaller integer intervals $\{t_1,\dots,t_2-1\},\{t_2,\dots,t_3-1\},\dots,\{t_B,\dots,t_{B+1}-1\}$ where $0=t_1<t_2<\dots<t_{B}<t_{B+1}=n+1$ are integers. Denoting the random test outcomes as $x\in\{0,1\}^N$, the score $f(x)$ is the value $i$ such that $\lVert x\rVert_1$, the number of successes in $x$, is contained in $\{t_i,\dots,t_{i+1}-1\}$. The goal is to efficiently compute a strategy that conducts tests so as to determine $f(x)$ at (approximately) minimum expected cost. We give a sample instance along with an optimal strategy in~\Cref{fig:example}.

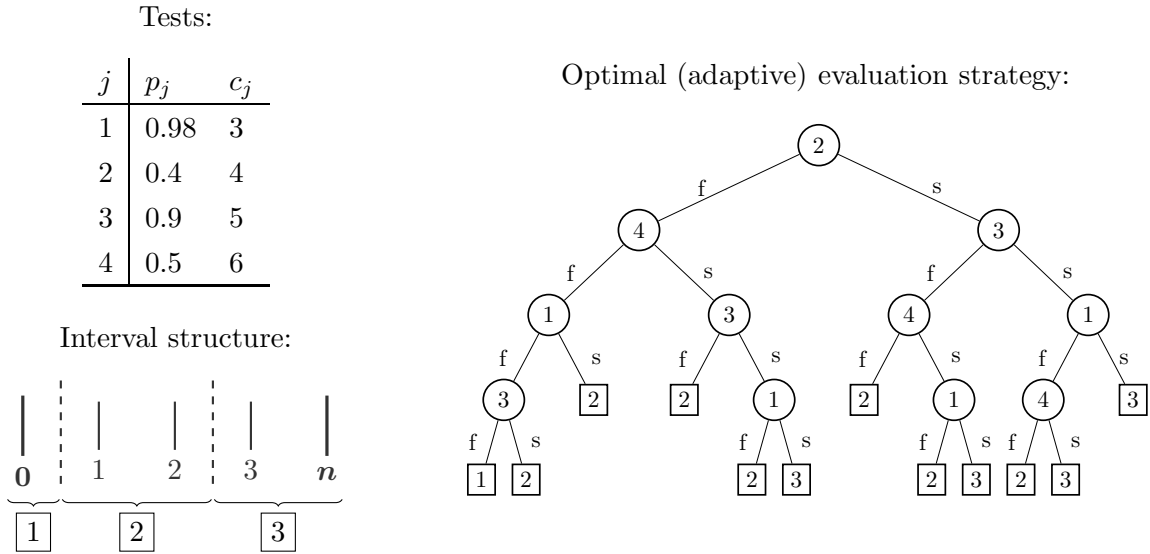
\begin{figure}

%\centering

\begin{minipage}{.3\textwidth}
    \centering
    Tests:\\
    \vspace{.4cm}
    \begin{tabular}{l | l l}
        $j$ & $p_j$ & $c_j$ \\
        \hline
        1 & 0.98 & 3 \\
        2 &  0.4 & 4 \\
        3 & 0.9 & 5 \\
        4 & 0.5 & 6 \\
    \hline
    \end{tabular}

    \vspace{.5cm}
    
    \centering
    
    Interval structure:\\
    \vspace{.4cm}
    \begin{tikzpicture}[yscale=.4]
        \tikzset{bar/.style={thick, color=black!80}}
        \tikzset{outerbar/.style={bar, very thick, node font=\boldmath}}
        \tikzset{conint/.style={pattern=north east lines}}
        
        %\draw[outerbar] (0,0) -- (4,0);
        
        \draw[outerbar] (0, 1) -- (0, -1) node[below]{$0$};
        \draw[outerbar] (4, 1) -- (4, -1) node[below]{$n$};
        \draw[bar] (1, .8) -- (1, -.8) node[below]{$1$};
        \draw[bar] (2, .8) -- (2, -.8) node[below]{$2$};
        \draw[bar] (3, .8) -- (3, -.8) node[below]{$3$};
        \draw[bar,dashed] (.5, 1.5) -- (.5, -2);
        \draw[bar,dashed] (2.5, 1.5) -- (2.5, -2);
        \begin{scope}[yshift = -2cm]
            \draw[decoration={brace,mirror,raise=5pt},decorate] (-.2,0) -- node[rectangle,draw=black,below=10pt] {1} (.48,0);
            \draw[decoration={brace,mirror,raise=5pt},decorate] (.52,0) -- node[rectangle,draw=black,below=10pt] {2} (2.48,0);
            \draw[decoration={brace,mirror,raise=5pt},decorate] (2.52,0) -- node[rectangle,draw=black,below=10pt] {3} (4.2,0);
        \end{scope}
    \end{tikzpicture}
\end{minipage}%
\begin{minipage}{.69\textwidth}
    \centering
    Optimal (adaptive) evaluation strategy:\\
    \vspace{.4cm}
    \scalebox{0.8}{
        \begin{forest}
            for tree={circle,draw, l sep=20pt, thick}
            [2 
                [4,edge label={node[midway,left] {f$\,\,\,\,$}} 
                    [1,edge label={node[midway,left] {f$\,\,\,$}} 
                        [3,edge label={node[midway,left] {f$\,\,\,$}} 
                            [1,edge label={node[midway,left] {f$\,\,$}}, rectangle ]
                            [2,edge label={node[midway,right] {$\,\,$s}}, rectangle]
                        ]
                        [2, rectangle, edge label={node[midway,right] {$\,\,\,$s}}
                            [x, white, no edge, rectangle]
                            [x, white, no edge, rectangle]
                        ]
                    ] 
                    [3,edge label={node[midway,right] {$\,\,\,$s}} 
                        [2, rectangle, edge label={node[midway,left] {f$\,\,\,$}} 
                            [x, white, no edge, rectangle]
                            [x, white, no edge, rectangle]
                        ]
                        [1,edge label={node[midway,right] {$\,\,\,$s}} 
                            [2, rectangle, edge label={node[midway,left] {f$\,\,$}}, rectangle]
                            [3, rectangle,edge label={node[midway,right] {$\,\,$s}}, rectangle]
                        ]
                    ]
                ]
                [3,edge label={node[midway,right] {$\,\,\,\,$s}} 
                    [4,edge label={node[midway,left] {f$\,\,\,$}} 
                        [2, rectangle, edge label={node[midway,left] {f$\,\,\,$}}, rectangle
                            [x, white, no edge, rectangle]
                            [x, white, no edge, rectangle]
                        ]
                        [1,edge label={node[midway,right] {$\,\,\,$s}} 
                            [2, rectangle, edge label={node[midway,left] {f$\,\,$}}, rectangle]
                            [3, rectangle,edge label={node[midway,right] {$\,\,$s}}, rectangle]
                        ]
                    ] 
                    [1,edge label={node[midway,right] {$\,\,\,$s}} 
                        [4,edge label={node[midway,left] {f$\,\,\,$}} 
                            [2, rectangle, edge label={node[midway,left] {f$\,\,$}}, rectangle]
                            [3, rectangle,edge label={node[midway,right] {$\,\,$s}}, rectangle]
                        ]
                        [3, rectangle, edge label={node[midway,right] {s$\,\,\,$}}, rectangle
                            [x, white, no edge, rectangle]
                            [x, white, no edge, rectangle]
                        ]
                    ] 
                ] 
            ]
        \end{forest}
    }

\end{minipage}

\caption{Sample instance with four tests and three intervals ($t_1 = 0, t_2 = 1, t_3 = 3, t_4 = 5$). The tree depicts the unique optimal adaptive evaluation strategy. A path to the left denotes \textbf{f}ailure, to the right \textbf{s}uccess of the test in the parent node. A circled node specifies the next test to be conducted; a boxed node corresponds to a determined score. Observe that the first test is not the cheapest; for some paths, the cheapest test is performed last.}
\label{fig:example}
\end{figure}

Gkenosis et al.\ provide two algorithms with approximation guarantees $O(\log n)$ and $O(B)$, respectively, as well as some constant-factor approximation algorithms and exact algorithms for special cases~\cite{GkenosisGHK:18,GkenosisGHK22,GrammelHKL:22}. A notable special case for which an exact adaptive algorithm is known are so called \KofN{} functions, the special case when $B=2$~\cite{BenDov81,SalloumB84,Salloum:79} (for a review of the fragmented literature leading to these results, we refer to Salloum and Breuer~\cite{SalloumB:97}).
In work concurrent to the present one, Ghuge et al.~\cite{GhugeGN:21} give a constant-factor approximation, however, without a special focus on simplicity and the achieved constant. We give a deeper review of the literature later on.

In general, strategies can be complicated exponential-size objects since any test to be conducted could depend on the outcomes of all previous tests; for computational issues arising from this, we refer to Ünlüyurt's survey~\cite[Section 4.1]{Unluyurt04}. Nevertheless, a class of simple and in practice desirable strategies are non-adaptive strategies, which conduct tests in a fixed order and stop once the function value of $f$ is uniquely determined. The adaptivity gap measures the loss in the approximation factor one incurs when restricting to non-adaptive strategies. Since the recent work of Ghuge et al.~\cite{GhugeGN:21} it is known that the adaptivity gap for Stochastic Score Classification is constant.

In this work, we give a simple non-adaptive approximation algorithm for SSC that achieves the relatively small guarantee of $3+2\sqrt{2}\approx 5.828$.
We also show that an even simpler algorithmic idea, which was previously studied~\cite{GkenosisGHK:18}, also yields an only slightly larger guarantee of $6$. The algorithm is also non-adaptive and conceivably the simplest constant-factor approximation algorithm.
Lastly, we lower bound the adaptivity gap by $\nicefrac32$, even for the unit-cost case, showing that the upper bound of Grammel et al.~\cite{GrammelHKL:22} for the unit-cost case is tight.

\paragraph*{Our Contribution.}

Let us first consider the following non-adaptive algorithm, which we call \TwoRR and which is a $2$-approximation for \KofN{} functions (see~\cite{GkenosisGHK:18} for the same result on a very similar algorithm; we later discuss the differences in more detail). Recall that \KofN{} functions can take two function values, namely 1 (less than $k$ successes) or 2 (at least $k$ successes). The algorithm is based on two non-adaptive (sub-)algorithms $\ALGSz$ and $\ALGSo$, which are optimal conditioned on the function value being 1 and 2, respectively \cite{BenDov81}. These algorithms can be interweaved by a round-robin scheme: In each step, conduct the next unconducted test of that sub-algorithm that has the smallest total cost among all sub-algorithms when conducting their respective next unconducted test. Then, conditioned on function value 1, the cost of both sub-algorithms can be charged (realization-wise) to the cost of $\ALGSz$ run on the instance by itself, the expected value of which is a lower bound on $\OPT$. A symmetric argument for the condition on function value~$2$ leads to the $2$-approximation.

The algorithms $\ALGSz$ and $\ALGSo$ are very simple: $\ALGSz$ simply conducts tests in increasing order of $c_j/(1-p_j)$~\cite{Butterworth:72}; symmetrically, $\ALGSo$ orders by $c_j/p_j$. In particular, note that neither $\ALGSz$ nor $\ALGSo$ depends on $k$.

Now consider the case of general SSC. One may be tempted to also use \TwoRR here with the following (flawed) analysis in mind: Conditioned on $\lVert x\rVert_1\in [t_i,t_{i+1})$, $\ALGSo$ would optimally verify that $\lVert x\rVert_1\geq t_i$, and $\ALGSz$ would optimally verify that $\lVert x\rVert_1<t_{i+1}$. One of the sub-algorithms may perform extra tests that are not needed, but its total cost could be charged to the cost of the other algorithm.

The fact that such an analysis does not work can be seen as follows. Consider $\ALGSo$.
This algorithm is only optimal for verifying that $\lVert x\rVert_1\geq t_1$ conditioned on that being indeed the case, not under the condition that $\lVert x\rVert_1\in [t_i,t_{i+1})$.
Indeed, under the latter condition, $\ALGSo$ may be far from optimal:
For the first test, let $p_1=1-\varepsilon$ and $c_1=1$; for all other tests $j\in\{2,\dots,n\}$, let $p_j=0.5$ and $c_j=0.5$.
Further let $t_2=1$, $t_3=2$, and $t_4=n+1$, and condition $\lVert x\rVert_1\in [t_2,t_3)$, that is, $\lVert x\rVert_1=1$.
As $p_1\rightarrow 1$ with $\varepsilon \rightarrow 0$, the conditional success probability of test $1$ converges to $1$, and that of all other tests to $0$.
Therefore, the fact that $\lVert x\rVert_1\geq 1$ can (almost) only be verified by conducting test $1$ at cost $1$; $\ALGSo$, however, performs tests $2,\dots,n$ first, paying cost $\Theta(n)$.

In fact, we will eventually show that \TwoRR \emph{is} a constant-factor approximation, but the correct analysis is rather intricate. We therefore present a slightly more complicated algorithm with a simpler analysis for an, in fact, slightly better approximation ratio first. 

First note that, when $i=1$ or $i=B$ and conditioned on $\lVert x\rVert_1\in [t_i,t_{i+1})$, $\ALGSz$ and $\ALGSo$ \emph{are} optimal for verifying that $\lVert x\rVert_1\geq t_i$ and $\lVert x\rVert_1<t_{i+1}$, respectively.
Also think of the outcome of a test as inducing a sub-instance with one test fewer.
Then the idea underlying our algorithm is: Also conduct \emph{cheap} tests to some degree, independent of their success probability, so as to, at preferably small cost, obtain a sub-instance (where the function to evaluate is given through $t'_1,\dots,t'_{B'+1}$) in which the number of successes is contained in the interval $[t'_1,t'_2)$ or $[t'_{B'},t'_{B'+1})$.
Indeed, it is not hard to see that the algorithm \ALGCh{} conducting tests in increasing order of costs obtains such a sub-instance by paying at most the total cost of the optimal strategy.

Of course, the algorithm is oblivious to the interval that $\lVert x\rVert_1$ is in and can therefore not identify the time at which the aforementioned sub-instance is obtained.
This issue can easily be addressed by applying the round-robin scheme to $\ALGSz$, $\ALGSo$, and the algorithm \ALGCh{}; we call the resulting algorithm $\TwRR$. The analysis then works again conditioned on $\lVert x\rVert_1\in [t_i,t_{i+1})$ and in two phases:
In the first phase, the sub-instance is obtained; and in the second phase the remaining interval boundary is verified.
Essentially, in the first phase we can charge the cost incurred by all sub-algorithms to \ALGCh{} (when run separately, not as part of \TwRR), which can be upper-bounded by $\OPT$.
For the second phase, suppose $\ALGSz$ optimally verifies the remaining interval boundary. Its cost (when run separately) can be upper-bounded by $\OPT$ (using the properties of a sub-instance), and the cost of all sub-algorithms can be charged to \ALGSz{}. By such an analysis, we can show that \TwRR is a $6$-approximation.
We also show that the approximation guarantee can be improved to $3+2\sqrt{2}\approx 5.828$ by a weighted version of the round-robin scheme.

The analysis of $\TwoRR$ is based on the fact that the tests that this algorithm conducts are somewhat similar to those that \TwRR conducts. More specifically, at any time, if neither \ALGSz{} nor \ALGSo{} has conducted some cheap test $j$, then at least one of them has only conducted tests whose cost are within a factor of $2$ of the cost of the cheap test.
Altogether, this allows us to prove an approximation guarantee of~$6$.

Since our algorithms are non-adaptive and we compare against a possibly adaptive optimum, our results imply an upper bound of $3+2\sqrt{2}$ on the adaptivity gap of Stochastic Score Classification.
To complement our upper bounds on the approximation ratio, we provide a family of instances that shows that the adaptivity gap cannot be smaller than $\nicefrac32$. Our instance happens to be a \KofN{} instance in which all tests have unit cost, which, together with the matching upper bound~\cite{GrammelHKL:22}, settles the adaptivity gap for this case to be $\nicefrac32$.

\paragraph*{Further Previous Work.} The Stochastic Score Classification problem is a specific Stochastic Function Evaluation (SFE) problem: Here, the function $f$ may instead be any known (succinctly represented) function mapping from $\{0,1\}^N$ to some value set, and the goal is again to determine its function value at minimum expected cost. For a (slightly outdated) in-depth review of the area, we refer to the survey by Ünlüyurt~\cite{Unluyurt04}.

We mention a few results on other special cases of the SFE problem closely related to the SSC problem. The aforementioned results on optimal adaptive strategies for \KofN{} functions~\cite{BenDov81,SalloumB84,Salloum:79} are generalized by Boros and Ünlüyurt~\cite{BorosU:99} to so-called double-regular functions (systems, in their terms), a specific Stochastic \emph{Boolean} Function Evaluation problem in which the variables fulfill certain dominance properties. Deshpande et al.~\cite{DeshpandeHK:16} give a 3-approximation for the weighted version of \KofN{} functions (linear-threshold formulas, in their terms) based on a so-called dual greedy algorithm for Stochastic Submodular Set Cover.

We discuss the known approximation algorithms for the general SSC problem in more detail. The $O(B)$-approximation algorithm~\cite{GkenosisGHK:18,GkenosisGHK22} is obtained by running the optimal adaptive strategy for the $t_2$-of-$n$, $t_3$-of-$n$, \dots,$t_B$-of-$n$ function, each of whose costs can be charged to the globally optimal strategy. The $O(\log n)$-approximation algorithms~\cite{GkenosisGHK:18,GkenosisGHK22}, on the other hand, employs techniques used by Deshpande et al.~\cite{DeshpandeHK:16} for Stochastic Submodular Set Cover.

Gkenosis et al.~\cite{GkenosisGHK:18} consider the unit-cost version of the problem and show a $4$-approximation algorithm. Up to performing the same test possibly twice (naturally, with identical outcome), their algorithm is the same as \TwoRR. In a recent paper, Grammel et al.~\cite{GrammelHKL:22} give a $2$-approximation algorithm.
In addition, the unanimous vote function has been considered, the special case where $B=3$, $t_2=1$, and $t_3=n$, i.e., $f(x)=1$ only if $\lVert x\rVert=0$, and $f(x)=3$ only if $\lVert x\rVert=n$. An optimal strategy is among $n$ strategies that can be efficiently compared~\cite{GkenosisGHK:18,GrammelHKL:22}.

Adaptivity gaps are also considered explicitly by Gkenosis et al.~\cite{GkenosisGHK:18,GkenosisGHK22,GrammelHKL:22}. They provide simple algorithms showing that the adaptivitiy gaps for the general case, \KofN{} functions, and unanimous vote functions are at most $2\cdot (B-1)$, $2$, and $2$, respectively. In the unit-cost case, these bounds can be improved to $2$, $\varphi=\nicefrac{(1+\sqrt{5})}2$, and $\nicefrac32$, respectively. Very recently, Hellerstein et al.~\cite{HellersteinKLW22} investigate adaptivity gaps for other SFE problems. To the best of our knowledge, lower bounds on the adaptivity gaps for SSC are not known.

It is an open question whether the SSC problem is NP-hard, and a polynomial-time algorithm is not even known for the unit-cost version of the problem. A Stochastic Boolean Function Evaluation problem that has the same status since about 50 years is the problem of evaluating stochastic read-once formulas (e.g.,~\cite{Unluyurt04,GreinerHJM06}). The weighted version of the SSC problem is clearly NP-hard because the deterministic version is already NP-hard due to the hardness of the knapsack problem.

\paragraph*{Note on Concurrent Work.} Concurrently to our work, Ghuge et al.~\cite{GhugeGN:21} have also announced an algorithm for SSC achieving a constant-factor approximation. For the $d$-dimensional weighted version of the problem, they present a polynomial-time $O(d^2 \log d)$-approximation algorithm.

In contrast to their work, our algorithm and analysis are arguably cleaner (also due to the fact that we consider a more restricted problem). Furthermore, we focus on obtaining a small approximation guarantee while Ghuge et al.\ only focus on showing that a constant-factor approximation exists. Their algorithm, when projected onto the (unweighted) SSC problem, can also be viewed as a (coarser form of) interweaving \ALGSz and \ALGSo.

In independent work, Liu~\cite{liu20226approximation} also showed a 6-approximation ratio for both 3R and 2R.

\paragraph*{Overview of the Paper.} In Section~\ref{sec:prelim}, we introduce needed concepts and results from the literature and introduce additional notation. In Section~\ref{sec:CFA}, we give the constant-factor approximation algorithms. In Section~\ref{sec:CFA:adap}, we give the lower-bound instance for the adaptivity gap. Finally, we state directions for future work in Section~\ref{sec:conclusion}.

\section{Preliminaries}
\label{sec:prelim}

We start by introducing some notation. An instance of SSC is given through probabilities $p_j$ and costs $c_j$ for all $j\in N$ as well as the interval boundaries $t_1, \ldots, t_{B+1}$. We overload the symbol for cost $c$ and denote by the random variable $\Cost{}{S}{\Inst}$ the overall cost of a strategy $S$ on instance \Inst{} (where the randomness is over the test outcomes).
Of particular interest is the strategy \OPT{}, an arbitrary but fixed optimal strategy for the instance at hand, i.e., $\Cost{}{\OPT}{\Inst}\leq\Cost{}{S}{\Inst}$ for any instance $I$ and all adaptive strategies $S$.
Slightly misusing notation, we sometimes use an algorithm and the strategy computed by it interchangeably.
In particular, we denote by $\Cost{}{\ALG}{\Inst}$ the overall cost of the strategy computed by $\ALG$ on instance $I$. 

We frequently use the three specific algorithms \ALGSz, \ALGSo, and \ALGCh{}, defined through permutations $\Sig{\Fail},\Sig{\Succ},\Sig{\Ch}: \{1,\dots,n\} \rightarrow N$. These permutations are chosen such that
$\nicefrac{c_{\SigVal{\Fail}{1}}}{(1-p_{\SigVal{\Fail}{1}})} \leq \nicefrac{c_{\SigVal{\Fail}{2}}}{(1-p_{\SigVal{\Fail}{2}})} \leq \ldots \leq \nicefrac{c_{\SigVal{\Fail}{n}}}{(1-p_{\SigVal{\Fail}{n}})}$, $\nicefrac{c_{\SigVal{\Succ}{1}}}{p_{\SigVal{\Succ}{1}}} \leq \nicefrac{c_{\SigVal{\Succ}{2}}}{p_{\SigVal{\Succ}{2}}} \leq \ldots \leq \nicefrac{c_{\SigVal{\Succ}{n}}}{p_{\SigVal{\Succ}{n}}}$, and
$c_{\SigVal{\Ch}{1}} \leq c_{\SigVal{\Ch}{2}} \leq \ldots \leq c_{\SigVal{\Ch}{n}}$.
Now \ALGSz{}, \ALGSo{}, and \ALGCh{} are defined to be the algorithms selecting tests by increasing values of $\Sig{\Fail}$, $\Sig{\Succ}$, and $\Sig{\Ch}$, respectively. 
The algorithms $\ALGSz$ and $\ALGSo$ have been analyzed in the literature under the condition that $f(x)=1$ and $f(x)=2$, respectively, when $B=2$.
\begin{lemma}[Ben-Dov~\cite{BenDov81}]
\label{lemm:smithproperty}
    Let \ALG be any adaptive strategy for an instance \Inst{} of SSC with $B = 2$. If $f(x)=1$ with nonzero probably, it holds that
    \begin{equation}
        \Exc{\Cost{}{\ALG}{\Inst}}{f\RB{x} = 1} \geq \Exc{\Cost{}{\ALGSz}{\Inst}}{f\RB{x} = 1}.
    \end{equation}
    Symmetrically, if $f(x)=2$ with nonzero probability, it holds that
    \begin{equation}
        \Exc{\Cost{}{\ALG}{\Inst}}{f\RB{x} = 2} \geq \Exc{\Cost{}{\ALGSo}{\Inst}}{f\RB{x} = 2}.
    \end{equation}
\end{lemma}
This observation was first stated by Ben-Dov~\cite{BenDov81}, together with a proof sketch. For a rigorous proof of \Cref{lemm:smithproperty} we refer to Boros and Ünlüyurt~\cite{BorosU:99}, who show a more general statement.

The round-robin scheme that Gkenosis et al.\ use for their algorithms for SSC is a natural algorithm-design technique previously used in related settings \cite{AllenHKU:15, CharikarFGKRS:02, KaplanKM:05} to run different sub-algorithms given by $\ALG_1,\dots,\ALG_k$ ``simultaneously'', where time is identified with cost. In contrast to Gkenosis et al., we consider a weighted version in this paper (see also Charikar et al.~\cite{CharikarFGKRS:02}), where $\alpha_1,\dots,\alpha_k>0$ are real-valued \emph{weights}, and our version explicitly skips any test that would be conducted for a second time.

Given $\ALG_1,\dots,\ALG_k$ and $\alpha_1,\dots,\alpha_k$, the round-robin scheme initializes $C_h$ to $0$ for all $h\in\{1,\dots,k\}$. As long as the function value is not determined, in each iteration, it defines $\delta_h$ to be the cost of the next test that $\ALG_h$ would perform and that was previously not performed by the round-robin scheme. Then the scheme chooses $h^\star\in\{1,\dots,k\}$ such that $\nicefrac{1}{\alpha_{h^\star}}\cdot(C_{h^\star}+\delta_{h^\star})$ is minimized, lets $\ALG_{h^\star}$ perform that test of cost $\delta_{h^\star}$, and updates $C_{h^\star}$ by adding $\delta_{h^\star}$ to it. In \Cref{alg:RR}, we summarize this procedure called \RdRn.

\begin{algorithm}[t]
\caption{$ \RdRn \RB{\ALG_1, \ldots, \ALG_k, \alpha_1, \ldots, \alpha_k}$}
\label{alg:RR}

$C_h \leftarrow 0 \,\,\, \forall \, h \in \CB{1, \ldots, k}$\;
\While{$f(x)$ not determined}{
    $\delta_h \leftarrow \textrm{ cost of next previously not performed test that }\ALG_h \textrm{ would perform } \forall \, h \in \{1, \ldots, k\}$\;\label{line:algRR:nextsel}
    choose $h^{\star} \in \argmin_{h \in \CB{1, \ldots, k}} \nicefrac{1}{\alpha_h}\cdot \RB{C_h + \delta_h}$\;
    perform the next test of $\ALG_{h^{\star}}$ that has previously not been performed\;
    $C_{h^\star} \leftarrow C_{h^\star} + \delta_{h^\star}$\;
}
\end{algorithm}

The difference between the version of the round-robin scheme that \GGHK{}\ consider and ours is that their version does not distinguish between tests that have been performed by the scheme already and tests that have not been performed. Therefore, their version may pick the same test in two different iterations of the main loop, and it would then also account for the cost of the test twice. Note that this may also change the order in which subsequent tests are performed. 

As a side note, we observe that the factor by which the costs of the two versions vary can be bounded by a function of $k,\alpha_1,\dots,\alpha_k$. We argue for $\alpha_1=\dots=\alpha_k$; the argument can easily be extended to the general case. In this case, the factor can be shown to be at most $k$, for both directions: The central observation is that, conditioned on some realization $x$, the total cost accumulated by a sub-algorithm in one of the versions can never be larger than the \emph{total} (actual) cost ${C}$, i.e., the cost across all sub-algorithms, in the other version. This is because, if the cost accumulated by a sub-algorithm $h^\star$ ever became larger than ${C}$, \emph{any} sub-algorithm $\ALG_h$ would exceed an accumulated cost of ${C}$ after adding the next respective test. Since the underlying order of tests for each sub-algorithm is the same in both versions, this means that $\ALG_h$ has already performed all tests that the corresponding sub-algorithm performed in the other version. This is a contradiction to the fact that the round-robin scheme has not yet terminated.

In the analysis of the second phase of our algorithms, using our version of the round-robin scheme allows us to restrict our attention to the sub-instance that remains in the beginning of the second phase. On the other hand, our analysis becomes arguably more difficult due to the fact that, in contrast to \GGHK{}'s version, we cannot use the inequality $\nicefrac{1}{\alpha_{h^\star}}\cdot C_{h^\star} \geq \nicefrac{1}{\alpha_{h}}\cdot C_{h}$ when the scheme has chosen $\ALG_{h^\star}$ to perform the previous test and $\ALG_h$ is any other algorithm used in the scheme: Indeed, consider $k=2$ and $\alpha_1=\alpha_2=1$. If both $\ALG_1$ and $\ALG_2$ would first test the tests $1$ and $2$ with $c_1=2$ and $c_2=1$ in this order, and all other tests have much larger cost, the inequality would not hold after the second iteration of our round-robin scheme.

We conclude this section by introducing some additional notation:
For an algorithm $\mathrm{RR} := \RdRn \RB{\ALG_1, \ldots, \ALG_k, \alpha_1, \ldots, \alpha_k}$ of the above type and step $\tau$ at which it has not terminated, let $\cost^{\mathrm{RR}}_\tau(\ALG_h,I)$ be the cost increment in iteration $\tau$ on instance $I$ due to $\ALG_h$, with value $c_{h^\star}$ if $h = h^\star$ and $0$ otherwise.

Observe that the set of SSC instances is closed under querying tests: 
Having queried a test $j$ in instance $\Inst$, the remainder of $\Inst$ is equivalent to an instance $\Inst'$ which has $j$ removed from the set of tests and its interval borders adapted depending on the outcome of $j$.
If $j$ was a failure, the limits do not have to be changed. If $j$ was a success, in $\Inst'$ we now need one success fewer to reach a score.
Interval limits that reach $0$ this way or exceed the number of remaining tests can be removed since the corresponding intervals cannot be reached any longer.
A score determined for $\Inst'$ can be used to reconstruct the score of $\Inst$.
We say that $\Inst'$ is a sub-instance of $\Inst$ if there exists a subset of tests with fixed outcomes such that $\Inst'$ is obtained after querying those tests.
We use the random variable $I^\ALG_{\tau}$ to specifically denote the sub-instance obtained from $\Inst$ after algorithm \ALG has queried $\tau$ steps.

\section{Constant-Factor Approximations for SSC}
\label{sec:CFA}

The goal of this section is to show that there is a simple polynomial-time non-adaptive $(3+2\sqrt{2})$-approximation algorithm for SSC and that the slightly larger approximation ratio of 6 can be achieved by an even simpler algorithm that also runs in polynomial time and is non-adaptive.

As we show, the algorithm $\TwoRR:=\RdRn(\ALGSz,\ALGSo,1,1)$ (as introduced by Gkenosis et al.~\cite{GkenosisGHK:18} using a different round-robin scheme) is the latter, even simpler, algorithm. As argued earlier, however, a naive analysis of this algorithm fails.
Therefore, we first give a cleaner two-phase analysis of the slightly more complicated algorithm $\TwRR:=\RdRn(\ALGSz,\ALGSo,\ALGCh,\ALPHSz,\ALPHSo,\ALPHCh)$ in Section~\ref{subsec:rr3}.
%This algorithm is $\TwRR:=\RdRn(\ALGSz,\ALGSo,\ALGCh,\ALPHSz,\ALPHSo,\ALPHCh)$. 
We later choose the weights $\ALPHSz$, $\ALPHSo$ and $\ALPHCh$ so as to minimize the approximation guarantee; note that, due to symmetry, it makes sense to choose $\ALPHSo=\ALPHSz$. Then, by relating \TwoRR to \TwRR, we give an analysis of \TwoRR in Section~\ref{subsec:rr2}.

\subsection{Weighted Three-Way Round Robin}
\label{subsec:rr3}

In this section we prove the following theorem.
\begin{theorem}
\label{thm:TwRR}
    The algorithm \TwRR{} with parameters $\ALPHSo=\ALPHSz=1$ and $\ALPHCh=\sqrt{2}$ is a $(3+2\sqrt{2})$-approximation algorithm for Stochastic Score Classification.
\end{theorem}

We show a slightly stronger statement, namely that, conditional on $f(x)=i$ for any $i\in\{1,\dots,B-1\}$, the expected cost of \TwRR is at most a factor of $3 + 2 \sqrt{2}$ larger than that of the optimal strategy.
The analysis works in two phases: Phase 1 ends with step $\tau_1$, the first step after which \TwRR has found at least $t_i$ successes or at least $n + 1 - t_{i + 1}$ failures.
In other words, $\tau_1$ is the first step after which the score lies in the leftmost or rightmost interval of the remaining restricted instance $I^\TwRR_{\tau_1}$.
This is depicted in \Cref{fig:phase-viz}.
Note that $\tau_1$ is a random variable of value possibly $0$ but at least $\min\{t_i,n + 1 - t_{i + 1}\}$.
Phase 2 starts at step $\tau_1+1$ and ends when the function value is determined, at step $\tau_2$.

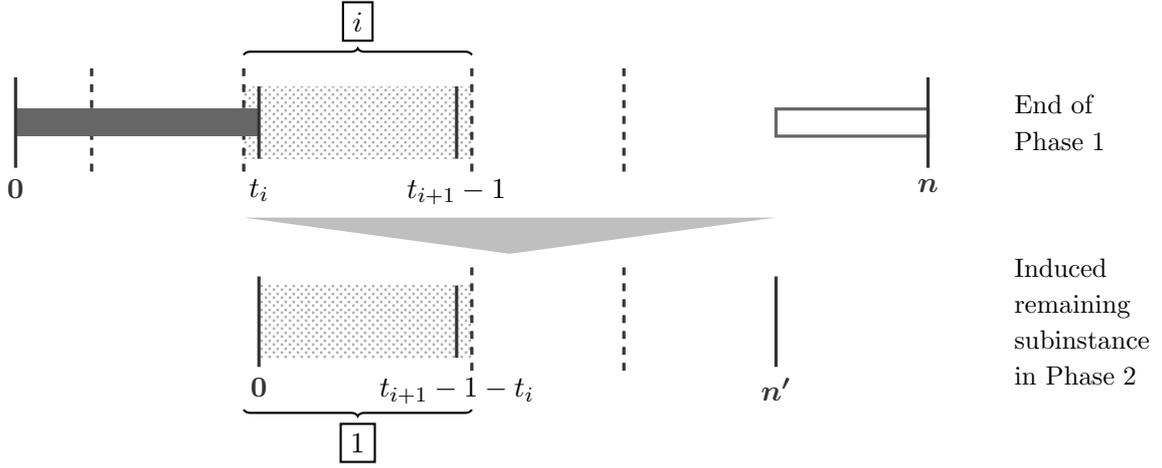
\begin{figure}
    \centering
    \begin{tikzpicture}[yscale=.6]
        \tikzset{bar/.style={very thick, color=black!80}}
        \tikzset{outerbar/.style={bar, very thick, node font=\boldmath}}
        \tikzset{conint/.style={pattern=north east lines}}
    
        \fill[pattern=crosshatch dots, pattern color=gray!60] (3, .8) rectangle (6, -.8);

        \draw[bar,dashed] (1, 1.2) -- (1, -1.2); 
        \draw[bar,dashed] (3, 1.2) -- (3, -1.2);
        \draw[bar,dashed,white] (3.2, 1.3) -- (3.2, -1.0) node[below,black]{$t_{i}$};
        \draw[bar,dashed,white] (5.8, 1.3) -- (5.8, -1.0) node[below,black]{$t_{i+1}-1$};
        \draw[bar,dashed] (6, 1.2) -- (6, -1.2);
        \draw[bar] (5.8, .8) -- (5.8, -.8);% node[below,fill=white]{$t_{i+1}-1$};
        \draw[bar,dashed] (8, 1.2) -- (8, -1.2);
        
        \fill[color=black!60] (0, .3) rectangle (3.2, -.3);
        \filldraw[color=black!60, very thick, fill=white] (10, .3) rectangle (12, -.3);
        \draw[bar] (3.2, .8) -- (3.2, -.8);% node[below,fill=white]{$t_i$};
        \draw[outerbar] (0, 1) -- (0, -1) node[below]{$0$};
        \draw[outerbar] (12, 1) -- (12, -1) node[below]{$n$};
        
        \draw (13, 0) node[align=left,anchor=west]{\small End of\\[-.2\baselineskip]\small Phase~$1$};
        
        \begin{scope}[shift={(0,-2)}]
            \draw[thick,decoration={brace,raise=5pt},decorate] (3,3.2) -- node[rectangle,draw=black,above=10pt] {$i$} (6,3.2);
            \draw [very thick, fill=gray!50, draw=none] (3,-.1) -- (6.5,-.9) -- (10,-.1);
        \end{scope}
        
        \begin{scope}[shift={(0,-4.4)}]
            \fill[pattern=crosshatch dots, pattern color=gray!60] (3.2, .8) rectangle (6, -.8);
    
            \draw[outerbar] (3.2, 1) -- (3.2, -1) node[below]{$0$};
            \draw[outerbar] (10, 1) -- (10, -1) node[below]{$n'$};
            \draw[bar,white] (5.8,1.3) -- (5.8,-1.0) node[below,black]{$t_{i+1}-1-t_i$};
            \draw[bar,dashed] (6,1.2) -- (6,-1.2);
            \draw[bar] (5.8,.8) -- (5.8,-.8);% node[below,black,fill=white]{$\tau_{i+1}-1-\tau_i$};
            %\draw[bar,dashed] (6, 1.2) -- (6, -1.2) node[below]{$\tau_{i^\ast+1} - \tau_{i^\ast}$};
            \draw[bar,dashed] (8, 1.2) -- (8, -1.2);
            \draw[thick,decoration={brace,mirror,raise=5pt},decorate] (3,-1.65) -- node[rectangle,draw=black,below=10pt] {$1$} (6,-1.65);
            
            \draw (13, 0) node[align=left,anchor=west]{\small Induced\\[-.2\baselineskip]\small remaining\\[-.2\baselineskip]\small subinstance\\[-.2\baselineskip]\small in Phase~$2$};
        \end{scope}
    
    \end{tikzpicture}
    \caption{Visualization of the end of phase~$1$ and the resulting subinstance.
    The dotted area shows the interval~$i$ containing the score according to the condition in \Cref{lemm:RR:Phase2}.
    The dark bar displays the successful tests performed so far, the white bar the failed tests.
    The number $n'$ of tests in the subinstance is obtained by subtracting the number of tests performed so far from $n$.}
    \label{fig:phase-viz}
\end{figure}

Intuitively, a good algorithm for phase 1 is $\ALGCh$, and the ``right'' algorithm for phase~2 is $\ALG_s$ where $s$ is a random variable whose value is ``fail'' if the score lies in the leftmost interval of $I^\TwRR_{\tau_1}$ and ``succ'' otherwise (equivalently, assuming phase~1 is nonempty, if the last test of phase~1 was a success or failure, respectively).

The following lemma bounds the cost caused by $\ALGSz$, $\ALGSo$, and $\ALGCh$ in phase~$1$. Note that the inequality in the statement always holds, not only in expectation conditioned on $f(x)=i$.

\begin{lemma}\label{lemm:RR:Phase1}
    For any $h \in \CB{\Fail, \Succ, \Ch}$, it holds that $$\frac{1}{\alpha_h}\cdot\sum_{\tau=1}^{\tau_1} \cost^\TwRR_\tau(\ALG_h,I)\leq \frac{1}{\ALPHCh}\cdot\Cost{}{\OPT}{\Inst}.$$
\end{lemma}
\begin{proof}
    First let $i:=f(x)$ and define $n_i:=t_i+(n + 1 - t_{i + 1})$, the minimum number of tests that need to be performed by any algorithm if the score lies in the interval $[t_i, t_{i+1})$. Hence the cost of the $n_i$ cheapest tests is a lower bound on $\Cost{}{\OPT}{\Inst}$:
    \begin{equation}\label{eq:phase1-opt}
        \Cost{}{\OPT}{\Inst}\geq \sum_{j=1}^{n_i} c_{\SigVal{\Ch}{j}}.
    \end{equation}
    
    On the other hand, consider some $h \in \CB{\Fail, \Succ, \Ch}$. If $\ALG_h$ was not chosen by $\TwRR$ to conduct any test in phase 1 (e.g., because phase 1 is empty), the lemma is trivially true. Otherwise, let $\hat{\tau}_1 \leq \tau_1$ be the last step in which $\TwRR$ chose $\ALG_h$ to conduct a test before the end of phase $1$.
        
    Now let $j^\star$ be the minimum value such that $\SigVal{\Ch}{j^\star}$ has not been conducted up to and including step $\hat{\tau}_1-1$. Note that it must hold that $j^\star\leq n_i$: Otherwise, at least $n_i$ tests would have already been conducted at that time, implying that already $t_i$ successes or $n+1-t_{i+1}$ failures would have been found and, consequently, that phase $1$ would have already ended. So, using the order in which $\ALGCh$ conducts tests, we have
	    \begin{equation}\label{eq:phase1-aux}
        \sum_{\tau=1}^{\hat{\tau}_1-1} \cost^\TwRR_\tau(\ALGCh,I)\leq \sum_{j=1}^{n_i}c_{\SigVal{\Ch}{j}}-c_{\SigVal{\Ch}{j^\star}}.
    \end{equation}
    For any $h \in \CB{\Fail, \Succ, \Ch}$, the way that \TwRR chooses the test to be conducted in step $\hat{\tau}_1$ now implies that
    \begin{equation}\label{eq:phase1-alg}
        \frac{1}{\alpha_h}\cdot\sum_{\tau=1}^{\hat{\tau}_1} \cost^\TwRR_\tau(\ALG_h,I)\leq\frac{1}{\ALPHCh}\cdot\left(\sum_{\tau=1}^{\hat{\tau}_1-1} \cost^\TwRR_\tau(\ALGCh,I)+c_{\SigVal{\Ch}{j^\star}}\right)\leq\frac{1}{\ALPHCh}\cdot\sum_{j=1}^{n_i}c_{\SigVal{\Ch}{j}},
    \end{equation}    
    where the second step follows from Inequality~\eqref{eq:phase1-aux}. The lemma now follows by combining Inequalities~\eqref{eq:phase1-opt} and~\eqref{eq:phase1-alg} together with the fact that $\ALG_h$ does not conduct tests in steps $\hat{\tau}_1+1$ to $\tau_1$.
\end{proof}

The next lemma is concerned with phase $2$. To be able to use it both for the analysis of \TwRR and that of \TwoRR, we prove a slightly more general statement, for a generic round-robin algorithm $$\GenRR:=\RdRn(\ALGSz,\ALGSo,\ALG_1,\dots,\ALG_{k},\ALPHSz,\ALPHSo,\alpha_1,\dots,\alpha_k),$$ where $k\geq 0$, $\ALG_1,\dots,\ALG_k$ are arbitrary algorithms, and $\alpha_1,\dots,\alpha_k>0$ are arbitrary weights.
We need definitions analogous to those of $s$, $\tau_1$, and $\tau_2$: Setting again $i=f(x)$, we let $\tau_1'$ be the first step after which $\GenRR$ has found at least $t_i$ successes or at least $n+1-t_{i+1}$ failures. Phase $1$ (of the execution of \GenRR) lasts from step $1$ up to and including step $\tau_1'$, and phase $2$ from step $\tau_1'+1$ up to and including step $\tau_2'$, the first step after which the value of $f$ is determined.
Finally, $s'$ is a random variable whose value is ``fail'' if the score lies in the leftmost interval of $I^\GenRR_{\tau_1'}$ and ``succ'' otherwise.

Consider $h\in\CB{\Fail, \Succ, 1,\dots,k}$ and denote the event that $\ALG_h$ actually performs a test in phase 2 by $\PhaseTwoActive{h}'$ (and the corresponding event for \TwRR, and \TwoRR later, by $\PhaseTwoActive{h}$). To complete the full analysis, we do not require an additional statement about $\ALG_h$ in phase 2 if $E_h$ does not occur. Therefore, the next lemma only refers to the part of the distribution corresponding to $E_h$. Informally, the statement is the following. Additionally conditioned on $f(x)=i$ and on $I^\GenRR_{\tau_1'}$ being any relevant sub-instance $I'$ of $I$, and up to the scaling by $\ALPHSz,\ALPHSo,\alpha_1,\dots,\alpha_k$, the expected cost caused by $\ALG_h$ going beyond the cost of $\ALG_{s'}$ in phase $1$ is upper-bounded by the expected cost of the optimal strategy. The proof idea is to first relate the cost of $\OPT$ to that of $\ALG_{s'}$, and then to use a similar argument as in the second part of the proof of Lemma~\ref{lemm:RR:Phase1}.

\begin{lemma}
    \label{lemm:RR:Phase2}
    Consider any instance $I'$, any $i\in\{1,\dots,B-1\}$, and any $h \in \CB{\Fail, \Succ, 1,\dots,k}$ such that $I^\GenRR_{\tau_1'}=I'\wedge f(x)=i\wedge\PhaseTwoActive{h}'$ with nonzero probability. Then we have
    \begin{align*}\label{eq:phase2lemm3}
        &\Ptyc{\PhaseTwoActive{h}'}{I^\GenRR_{\tau_1'}=I'\wedge f(x)=i}\\
        \cdot\;&\Exc{\frac{1}{\alpha_{h}}\cdot\sum_{\tau=1}^{\tau_2'}\cost^{\GenRR}_\tau(\ALG_{h},I)-\frac{1}{\alpha_{s'}}\cdot\sum_{\tau=1}^{\tau_1'}\cost^{\GenRR}_\tau(\ALG_{s'},I)}{I^\GenRR_{\tau_1'}=I'\wedge f(x)=i\wedge\PhaseTwoActive{h}'}\\
        \leq \phantom{\cdot\;}&\Exc{\frac{1}{\alpha_{s'}}\cdot\Cost{}{\OPT}{\Inst}}{I^\GenRR_{\tau_1'}=I'\wedge f(x)=i}.
    \end{align*}
\end{lemma}

\begin{proof}
    First note that
    \begin{equation}\label{eq:phase2-1}
        \Exc{\cost(\OPT,I)}{I^\GenRR_{\tau_1'}=I'\wedge f(x)=i}\geq\Exc{\cost(\OPT',I')}{I^\GenRR_{\tau_1'}=I'\wedge f(x)=i}
    \end{equation}
    where $\OPT'$ denotes the optimal strategy for $I'$ conditioned on $I^\GenRR_{\tau_1'}=I'\wedge f(x)=i$. Note that, while $\cost(\OPT',I')$ does not refer to $I$, the condition $I^\GenRR_{\tau_1'}=I'$ (which includes $I$) is necessary because $f(x)=i$ refers to $I$. The inequality holds because, given that $I'$ is a sub-instance of $I$, any algorithm for $I$, in particular $\OPT$, induces an algorithm for $I'$. Here, all tests, in particular those that appear in $I$ but not in $I'$, are evaluated according to their success probabilities conditional on $I^\GenRR_{\tau_1'}=I'\wedge f(x)=i$. Note that this condition may introduce (unproblematic) correlations between the successes of these tests not in $I'$. This transformation leads to a realization-wise non-increased cost caused by the new algorithm on $I'$ compared to the original algorithm on $I$ because the costs of the additional tests not appearing in $I'$ are ignored.
    
    Next observe that $I^\GenRR_{\tau_1'}$ does not contain tests that $\GenRR$ conducts in phase $1$. Consequently, the condition $I^\GenRR_{\tau_1'}=I'$ does \emph{not} refer to the tests contained in $I'$ and, conditioned on $I^\GenRR_{\tau_1'}=I'$, any test $j$ that appears in $I'$ is still successful independently with probability $p_j$.
    Therefore, according to Lemma~\ref{lemm:smithproperty}, $\ALG_{s'}$ is the optimal algorithm for $I'$ under the additional condition that $f(x)=i$:
    \begin{equation}\label{eq:phase2-2}
        \Exc{\cost(\OPT',I')}{I^\GenRR_{\tau_1'}=I'\wedge f(x)=i}=\Exc{\cost(\ALG_{s'},I')}{I^\GenRR_{\tau_1'}=I'\wedge f(x)=i}.
    \end{equation}

    Now consider some $h \in \CB{\Fail, \Succ, 1,\dots,k}$.
    Assuming that $\PhaseTwoActive{h}'$ occurs, we define step $\hat{\tau}_2$ with $\tau'_1 < \hat{\tau}_2 \leq \tau'_2$ such that it is the last step in which \GenRR chose $\ALG_h$ to conduct a test. Indeed, $\hat{\tau}_2$ is well defined if $\PhaseTwoActive{h}'$ occurs; we do not use it otherwise.
    
    Suppose for a moment that we let $\ALG_{s'}$ (rather than one of the other $k+1$ algorithms) conduct its next test in every step from step $\tau_1'+1$ onward, until $f$ is evaluated.
    Let $\ell_1,\dots,\ell_m$ be the tests that $\ALG_{s'}$ would then conduct in this order.  
    Moving back to the tests that $\GenRR$ actually performs, by the definition of $\tau_2'$, there must exist a test among $\ell_1,\dots,\ell_m$ not yet conducted up to and including step $\tau_2'-1$.
    Therefore, we can define $j^\star\in\{1,\dots,m\}$ to be minimal such that $\ell_{j^\star}$ has not been conducted up to and including step $\hat{\tau}_2-1 \leq \tau_2'-1$.
    Then, using the order in which $\ALG_{s'}$ conducts tests, we obtain
    \begin{equation}\label{eq:phase2-3}
        \sum_{\tau=\tau_1'+1}^{\hat{\tau}_2-1}\cost^\GenRR_\tau(\ALG_{s'},I)\leq \sum_{j=1}^{m} c_{\ell_j}-c_{\ell_{j^\star}}\, .
    \end{equation}
    The way that $\GenRR$ chooses the test to be conducted in step $\hat{\tau}_2$ now implies
    \begin{align}
        \frac{1}{\alpha_h}\cdot \sum_{\tau=1}^{\hat{\tau}_2}\cost^\GenRR_\tau(\ALG_h,I)&\leq \frac{1}{\alpha_{s'}}\cdot\left(\sum_{\tau=1}^{\hat{\tau}_2-1}\cost^\GenRR_\tau(\ALG_{s'},I)+c_{\ell_{j^\star}}\right)\notag\\
        &\leq \frac{1}{\alpha_{s'}}\cdot\left(\sum_{\tau=1}^{\tau_1'}\cost^\GenRR_\tau(\ALG_{s'},I)+\sum_{j=1}^{m} c_{\ell_j}\right)\notag\\
        &= \frac{1}{\alpha_{s'}}\cdot\left(\sum_{\tau=1}^{\tau_1'}\cost^\GenRR_\tau(\ALG_{s'},I)+\cost(\ALG_{s'},I^\GenRR_{\tau_1'})\right),\label{eq:phase2-3a}
    \end{align}
    where the second inequality follows from Inequality~\eqref{eq:phase2-3}, and the equality follows from the definition of $\ell_1,\dots,\ell_m$ (assuming that the outcomes of corresponding tests in $I$ and $I^\GenRR_{\tau_1'}$ are coupled in the sense that they are always identical). Rearranging then yields
    \begin{equation*}
        \frac{1}{\alpha_h}\cdot \sum_{\tau=1}^{\hat{\tau}_2}\cost^\GenRR_\tau(\ALG_h,I)-\frac{1}{\alpha_{s'}}\cdot\sum_{\tau=1}^{\tau_1'}\cost^\GenRR_\tau(\ALG_{s'},I)\leq\frac{1}{\alpha_{s'}}\cdot\cost(\ALG_{s'},I^\GenRR_{\tau_1'}).
    \end{equation*}
    By taking the expected value of this inequality (which holds realization-wise), we obtain
    \begin{align*}
        &\Ptyc{\PhaseTwoActive{h}'}{I^\GenRR_{\tau_1'}=I'\wedge f(x)=i}\\
        \cdot\;&\Exc{\frac{1}{\alpha_{h}}\cdot\sum_{\tau=1}^{\hat{\tau}_2}\cost^{\GenRR}_\tau(\ALG_{h},I)-\frac{1}{\alpha_{s'}}\cdot\sum_{\tau=1}^{\tau_1'}\cost^{\GenRR}_\tau(\ALG_{s'},I)}{I^\GenRR_{\tau_1'}=I'\wedge f(x)=i\wedge\PhaseTwoActive{h}'}\\
        \leq\phantom{\cdot\;}& \Ptyc{\PhaseTwoActive{h}'}{I^\GenRR_{\tau_1'}=I'\wedge f(x)=i}\cdot \Exc{\frac{1}{\alpha_{s'}}\cdot\cost(\ALG_{s'},I^\GenRR_{\tau_1'})}{I^\GenRR_{\tau_1'}=I'\wedge f(x)=i\wedge \PhaseTwoActive{h}'}\\
        \leq\phantom{\cdot\;}& \Exc{\frac{1}{\alpha_{s'}}\cdot\cost(\ALG_{s'},I')}{I^\GenRR_{\tau_1'}=I'\wedge f(x)=i},
    \end{align*}
    where the second step is due to the law of total expectation and the non-negativity of the cost function.
This, together with the fact that $\ALG_h$ does not conduct tests after step $\hat{\tau}_2$ as well as Inequalities~\eqref{eq:phase2-2} and~\eqref{eq:phase2-1}, implies the claim.
\end{proof}

Now the theorem follows quite easily from combining Lemma~\ref{lemm:RR:Phase1} and Lemma~\ref{lemm:RR:Phase2}.

\begin{proof}[Proof of Theorem~\ref{thm:TwRR}.]
    Consider instance $I'$ and $i\in\{1,\dots,B-1\}$ such that $I^\TwRR_{\tau_1}=I'\wedge f(x)=i$ with nonzero probability. Further, let $h \in \CB{\Fail, \Succ, \Ch}$.
    We consider whether each of $\PhaseTwoActive{h}$ and $\PhaseTwoInactive{h}$ occurs with nonzero probability under the condition $I^\TwRR_{\tau_1}=I'\wedge f(x)=i$.
    In case the first event occurs with nonzero probability, we use Lemmas~\ref{lemm:RR:Phase1} and~\ref{lemm:RR:Phase2}; for the second event, we only need Lemma~\ref{lemm:RR:Phase1}. We then combine the inequalities for the two events, in turn combine the resulting inequalities across all $I'$, $i$, and $h$, and finally optimize $\ALPHSz$, $\ALPHSo$, and $\ALPHCh$.
    
    We start with the case $\Ptyc{\PhaseTwoActive{h}}{I^\TwRR_{\tau_1}=I'\wedge f(x)=i} > 0$. By first applying \Cref{lemm:RR:Phase2} with \TwRR taking the role of \GenRR (meaning that $s$, $\tau_1$, $\tau_2$, and $E_h$ take the role of $s'$, $\tau_1'$, $\tau_2'$, and $E_h'$, respectively) and then applying Lemma~\ref{lemm:RR:Phase1} (which even holds realization-wise), we obtain
    \begin{align}
        \Ptyc{\PhaseTwoActive{h}}{I^\TwRR_{\tau_1}=I'\wedge f(x)=i}\;\cdot\;&\Exc{\sum_{\tau=1}^{\tau_2}\cost^{\TwRR}_\tau(\ALG_{h},I)}{I^\TwRR_{\tau_1}=I'\wedge f(x)=i\wedge\PhaseTwoActive{h}}\notag\\
        \leq  \Ptyc{\PhaseTwoActive{h}}{I^\TwRR_{\tau_1}=I'\wedge f(x)=i}\;\cdot\;&\Exc{\frac{\alpha_{h}}{\alpha_{s}}\cdot\sum_{\tau=1}^{\tau_1}\cost^{\TwRR}_\tau(\ALG_{s},I)}{I^\TwRR_{\tau_1}=I'\wedge f(x)=i\wedge\PhaseTwoActive{h}} \notag\\
        + \;& \Exc{\frac{\alpha_{h}}{\alpha_{s}}\cdot\Cost{}{\OPT}{\Inst}}{I^\TwRR_{\tau_1}=I'\wedge f(x)=i}\notag\\
        \leq \Ptyc{\PhaseTwoActive{h}}{I^\TwRR_{\tau_1}=I'\wedge f(x)=i} \;\cdot\;&\Exc{\frac{\alpha_{h}}{\ALPHCh}\cdot\Cost{}{\OPT}{\Inst}}{I^\TwRR_{\tau_1}=I'\wedge f(x)=i\wedge\PhaseTwoActive{h}}\notag\\
        +\;& \Exc{\frac{\alpha_{h}}{\alpha_{s}}\cdot\Cost{}{\OPT}{\Inst}}{I^\TwRR_{\tau_1}=I'\wedge f(x)=i}.\label{eq:rr3-thm-phasetwoactive}
    \end{align}
    For the case that $\Ptyc{\PhaseTwoInactive{h}}{I^\TwRR_{\tau_1}=I'\wedge f(x)=i} > 0$, we only apply Lemma~\ref{lemm:RR:Phase1} and obtain
    \begin{align}
    &\Ptyc{\PhaseTwoInactive{h}}{I^\TwRR_{\tau_1}=I'\wedge f(x)=i}\cdot\Exc{\sum_{\tau=1}^{\tau_2}\cost^{\TwRR}_\tau(\ALG_{h},I)}{I^\TwRR_{\tau_1}=I'\wedge f(x)=i\wedge\PhaseTwoInactive{h}}\notag\\
    = \;&\Ptyc{\PhaseTwoInactive{h}}{{I^\TwRR_{\tau_1}=I'\wedge f(x)=i}}\cdot\Exc{\sum_{\tau=1}^{\tau_1}\cost^{\TwRR}_\tau(\ALG_{h},I)}{I^\TwRR_{\tau_1}=I'\wedge f(x)=i\wedge\PhaseTwoInactive{h}}\notag\\
    \leq \;&\Ptyc{\PhaseTwoInactive{h}}{{I^\TwRR_{\tau_1}=I'\wedge f(x)=i}}\cdot\Exc{\frac{\alpha_{h}}{\ALPHCh}\cdot\Cost{}{\OPT}{\Inst}}{I^\TwRR_{\tau_1}=I'\wedge f(x)=i\wedge\PhaseTwoInactive{h}}.\label{eq:rr3-thm-phasetwoinactive}
    \end{align}
    
    For the next step, define $\mathcal{E}$ to be the set of events from $\{\PhaseTwoActive{h},\PhaseTwoInactive{h}\}$ with nonzero probability conditional on $I^\TwRR_{\tau_1}=I'\wedge f(x)=i$. Then, by the law of total expectation, we obtain
    \begin{align*}
    &\Exc{\sum_{\tau=1}^{\tau_2}\cost^{\TwRR}_\tau(\ALG_{h},I)}{I^\TwRR_{\tau_1}=I'\wedge f(x)=i}\\
        = \;& \sum_{E\in\mathcal{E}}\Ptyc{E}{I^\TwRR_{\tau_1}=I'\wedge f(x)=i}\cdot\Exc{\sum_{\tau=1}^{\tau_2}\cost^{\TwRR}_\tau(\ALG_{h},I)}{I^\TwRR_{\tau_1}=I'\wedge f(x)=i\wedge E}\\
        \leq \;&\phantom{{}+{}}\Exc{\frac{\alpha_{h}}{\ALPHCh}\cdot\Cost{}{\OPT}{\Inst}}{I^\TwRR_{\tau_1}=I'\wedge f(x)=i}\\
        &+\Exc{\frac{\alpha_{h}}{\alpha_{s}}\cdot\Cost{}{\OPT}{\Inst}}{I^\TwRR_{\tau_1}=I'\wedge f(x)=i}.
    \end{align*}
    In the second step, we use Equations~\eqref{eq:rr3-thm-phasetwoactive} (if $\PhaseTwoActive{h}\in\mathcal{E}$) and~\eqref{eq:rr3-thm-phasetwoinactive} (if $\PhaseTwoInactive{h}\in\mathcal{E}$) as well as the law of total expectation again.
    
    By applying the law of total expectation over all $I'$ and $i\in\{1,\dots,B-1\}$ such that $I^\TwRR_{\tau_1}=I'$ and $f(x)=i$ with nonzero probability and rearranging, we obtain from the previous inequality that
    \begin{equation}
        \Ex{\sum_{\tau=1}^{\tau_2}\cost^\TwRR_\tau\RB{\ALG_h,I}} \leq \alpha_h \cdot \Ex{\RB{\frac{1}{\ALPHCh} + \frac{1}{\alpha_s}} \cdot\Cost{}{\OPT}{\Inst}} \, .\label{eq:eq:rr3-thm-alphas}
    \end{equation}
    Now using that
    \begin{equation*}
        \cost(\TwRR,I)=\sum_{\tau=1}^{\tau_2}\cost^\TwRR_\tau(\ALGSz,I)+\sum_{\tau=1}^{\tau_2}\cost^\TwRR_\tau(\ALGSo,I)+\sum_{\tau=1}^{\tau_2}\cost^\TwRR_\tau(\ALGCh,I),    
    \end{equation*}
    we get from Inequality~\eqref{eq:eq:rr3-thm-alphas} that
    \begin{align*}
        \Ex{\cost(\TwRR,I)}&\leq(\ALPHSz+\ALPHSo+\ALPHCh)\cdot\Ex{\left(\frac{1}{\ALPHCh}+\frac{1}{\alpha_s}\right)\cdot\Cost{}{\OPT}{\Inst}}\\
        &=(2\ALPHSz+\ALPHCh)\cdot\left(\frac{1}{\ALPHCh}+\frac{1}{\ALPHSz}\right)\cdot\Ex{\Cost{}{\OPT}{\Inst}},
    \end{align*}
    where we set $\ALPHSz=\ALPHSo\;(=\alpha_s)$ for the equality. Minimizing the expression in front of the expected value on the right-hand side yields a minimum of $3+2\sqrt{2}$ at $\ALPHCh/\ALPHSz=\sqrt{2}$, completing the proof.
\end{proof}

Finally, we remark that the previous arguments, by simply setting $\ALPHSz=\ALPHSo=\ALPHCh$ rather than optimizing, show that the resulting unweighted round-robin algorithm is a $6$-approximation algorithm. In the following, we show that the same approximation ratio can be achieved without $\ALGCh$.

\subsection{Two-Way Round Robin}
\label{subsec:rr2}

In this section, we investigate the simpler algorithm $\TwoRR=\RdRn \RB{\ALGSz, \ALGSo, 1, 1}$ based on the insights and analysis from the previous section.
We prove that this natural strategy obtained by interweaving \ALGSz{} and \ALGSo{}---the optimal algorithms for $B = 2$ when the score is known---achieves in fact a similarly small approximation guarantee.

\begin{theorem}
    \label{thm:TwoRR}
    The algorithm \TwoRR{} is a $6$-approximation algorithm for Stochastic Score Classification.
\end{theorem}

Our analysis works again in two phases, which we define analogously to \Cref{subsec:rr3} but with respect to \TwoRR instead of \TwRR.
We also define the random variables $\tau_1$, $\tau_2$, and $s$ as well as the event $E_h$ as before, again with respect to \TwoRR.
While we can directly apply \Cref{lemm:RR:Phase2} for the second phase here, let us look more closely at how the cost in phase~1 can be bounded.

For \TwRR, we have introduced sub-algorithm \ALGCh{} as an additional component. This makes bounding the cost in phase~1 more or less straightforward since it provides us with a direct connection between the cost of \TwRR and the cost of cheap items, which were used as a lower bound on the cost of the optimal strategy.

Intuitively, one would expect that cheap tests are also attractive for \TwoRR (depending on their probability of success, either for \ALGSo{} or \ALGSz{} or both) and thus will be inspected early on.
In this section, we show that by looking at inspected tests from the right angle and cleverly assigning them to sets for which we can bound the cost, \TwoRR actually has a small constant approximation ratio.

Unfortunately, there is no direct bound between the cost of the $i$-th cheapest overall test and the cost of the $i$-th test selected by \ALGSo{} or \ALGSz{}:
For example, \ALGSo{} prefers a test with cost $1 - 2 \varepsilon$ and probability $1-\varepsilon$ over a test with cost and probability equal to $\varepsilon$.

However, if we compare the cost of the $i$-th test selected by \ALGSo{} with the cost of the $i$-th cheapest test among those with a lower bound on success probability $p_j$---let's pick $\nicefrac{1}{2}$ for \mbox{symmetry---,} we can indeed establish a bound:
Consider the first test $j$ selected by \ALGSo{} (i.e., $j = \SigVal{\Succ}{1}$) and the cheapest test $j'$ among tests with $p_{j'} \geq \frac12$.
%If $j = j'$, we have a direct correspondence.
%If $p_j < \nicefrac12$, then $j$ was selected by \ALGSo{} because $\nicefrac{c_j}{p_j} \leq \nicefrac{c_{j'}}{p_{j'}}$; in particular, $c_j < c_{j'}$.
We obtain $c_j \leq \nicefrac{p_j}{p_{j'}}\cdot c_{j'} \leq 2 c_{j'}$ since $p_j \leq 1$ and $p_{j'} \geq \nicefrac12$.
All in all, $j$ costs at most twice as much as $j'$.
The same principle works inductively and, of course, symmetrically for \ALGSz{} with an upper bound of $\nicefrac{1}{2}$ on $p_j$.
We establish a generalized version of this correspondence as a building block for the analysis of phase~1 in the following lemma.

Towards this, we partition the set of tests $N$ in two \enquote{half spaces} $\Hsp_\HSfail := \CB{1 \leq j \leq n : p_j < \nicefrac12}$ and $\Hsp_\HSsucc := \CB{1 \leq j \leq n : p_j \geq \nicefrac12}$ according to their success probability.
Informally, the lemma states the following: For a given set $F$ of tests contained in $\Hsp_\HSsucc$, if we let $\ALGSo$ pick tests from the union of $F$ and an arbitrary set of further tests, then the total cost of the first $k$ tests chosen by $\ALGSo$ is at most twice the cost of the cheapest $k$ tests from $F$.
This level of generality is useful later on in the proof of \Cref{thm:TwoRR}.

\begin{lemma}\label{lemm:Alg2:helper}
    For $u \in \CB{\Fail, \Succ}$, let $F \subseteq \Hsp_u$ and $G \subseteq N$ such that $F \subseteq G$.
    Let $G_{\lvert F \rvert}$ be the subset of $G$ containing the $\lvert F \rvert$ tests with smallest values of $\Sig{u}^{-1}$.
    Then, 
    \begin{equation*}
        \sum_{j \in G_{\lvert F \rvert}} c_j \leq  2 \cdot \sum_{j \in F} c_j \, .
    \end{equation*}
\end{lemma}
\begin{proof}
    Due to symmetry, it suffices to show the statement for $u = \Succ$. For $k\in\{0,\dots,|F|\}$, let $F_k$ be the subset of $F$ containing the $k$ tests with smallest value of $\Sig{\Ch}^{-1}$, and $G_k$ be the subset of $G$ containing the $k$ tests with smallest value of $\Sig{\Succ}^{-1}$. We show that, for all such $k$, 
    \begin{equation*}\label{eqn:twobound}
        \sum_{j \in G_k} c_j \leq  2 \cdot \sum_{j \in F_k} c_j \, .
    \end{equation*}
    We proceed by induction on $k$. 
    
    For $k = 0$, the statement clearly holds.    
    Now consider $k>0$, and assume that we have established the hypothesis for all smaller values. Let $j_F$ and $j_G$ be the tests such that $F_k = F_{k-1} \cup \CB{j_F}$ and $G_k = G_{k-1} \cup \CB{j_G}$, i.e., the tests with the $k$-th lowest value of $\Sig{\Ch}^{-1}$ and $\Sig{\Succ}^{-1}$, in $F_k$ and $G_k$, respectively.
    While it is possible that $j_F \in G_{k-1}$, there exists a test $j_F' \in F_k\setminus G_{k-1}$ with $c_{j_F'} \leq c_{j_F}$.
    By choice of $j_F'$ and definition of $G_{k}$, we have $\nicefrac{c_{j_G}}{p_{j_G}} \leq \nicefrac{c_{j_F'}}{p_{j_F'}}$.
    Since $p_{j_G} \leq 1$ and $p_{j_F'} \geq \nicefrac12$, we obtain $c_{j_G} \leq \nicefrac{p_{j_G}}{p_{j_F'}} \cdot c_{j_F'} \leq 2 c_{j_F'} \leq 2 c_{j_F}$.
    Combining this with the induction hypothesis yields
    \begin{equation*}
        \sum_{j \in G_k} c_j = c_{j_G} + \sum_{j \in G_{k-1}} c_j \,\, \leq \,\, 2 c_{j_F} + 2 \cdot \sum_{j \in F_{k-1}} c_j = 2 \cdot \sum_{j \in F_k} c_j \, .
    \end{equation*}
    This completes the induction step and therefore the proof.
\end{proof}

One (flawed) idea to proceed for bounding the cost in phase~1 would be the following, again using the notation $n_i=t_i+(n + 1 - t_{i + 1})$:
Partition the set \ChSet of $n_i$ \emph{cheapest} tests by success probability:
$\ChSet_\HSfail := \ChSet \cap \Hsp_\HSfail$ and $\ChSet_\HSsucc := \ChSet \cap \Hsp_\HSsucc$ for $\ChSet := \CB{\SigVal{\Ch}{1}, \SigVal{\Ch}{2}, \ldots, \SigVal{\Ch}{n_i}}$.

Then, by \Cref{lemm:Alg2:helper}, we can bound the cost of the first $|\ChSet_\HSsucc|$ tests from $\ALGSo$ with a factor of two against $\sum_{j \in \ChSet_\HSsucc} c_j$ and of the first $|\ChSet_\HSfail|$ tests from $\ALGSz$ with a factor of two against $\sum_{j \in \ChSet_\HSfail} c_j$. If we apply this bound at the first time when the $|\ChSet_\HSsucc|$ tests with smallest value of $\sigma_{\Succ}^{-1}$ and the $|\ChSet_\HSfail|$ tests with smallest value of $\sigma_{\Fail}^{-1}$ have been inspected, in combination with how the round-robin scheme is defined, we obtain an upper bound of $4 \cdot \sum_{j \in \ChSet_\HSsucc} c_j$ or $4 \cdot \sum_{j \in \ChSet_\HSfail} c_j \leq 4 \cdot \sum_{j \in \ChSet} c_j \leq 4\cdot\Cost{}{\OPT}{\Inst}$ on the total cost so far.

The problem is that we have no guarantee that at that point we have already reached the end of phase~1.
Some tests may have appeared in the prefixes of both sequences, so that we haven't reached a count of $n_i$ unique inspected tests yet.
Indeed, cheap tests are attractive for both \ALGSo and \ALGSz if the probability is around $\frac12$.
We therefore see that simply counting the number of tests of each (sub-)algorithm in this way is not sufficient. Instead, we will develop a way of counting unique inspected tests while at the same time still being able to bound the cost of all (i.e., the uniquely counted as well as the other) tests. We need to define a number of sets, visualized in \Cref{fig:Alg2:set-defs}.

% definitions needed to have pgf plotmarks available in normal text
\newbox\aMark
\setbox\aMark\hbox{\begin{pgfpicture}\pgfuseplotmark{x}\end{pgfpicture}}
\newbox\sMark
\setbox\sMark\hbox{\begin{pgfpicture}\pgfuseplotmark{square}\end{pgfpicture}}
\newbox\rMark
\setbox\rMark\hbox{\begin{pgfpicture}\pgfuseplotmark{o}\end{pgfpicture}}
\newbox\bMark
\setbox\bMark\hbox{\begin{pgfpicture}\pgfuseplotmark{+}\end{pgfpicture}}

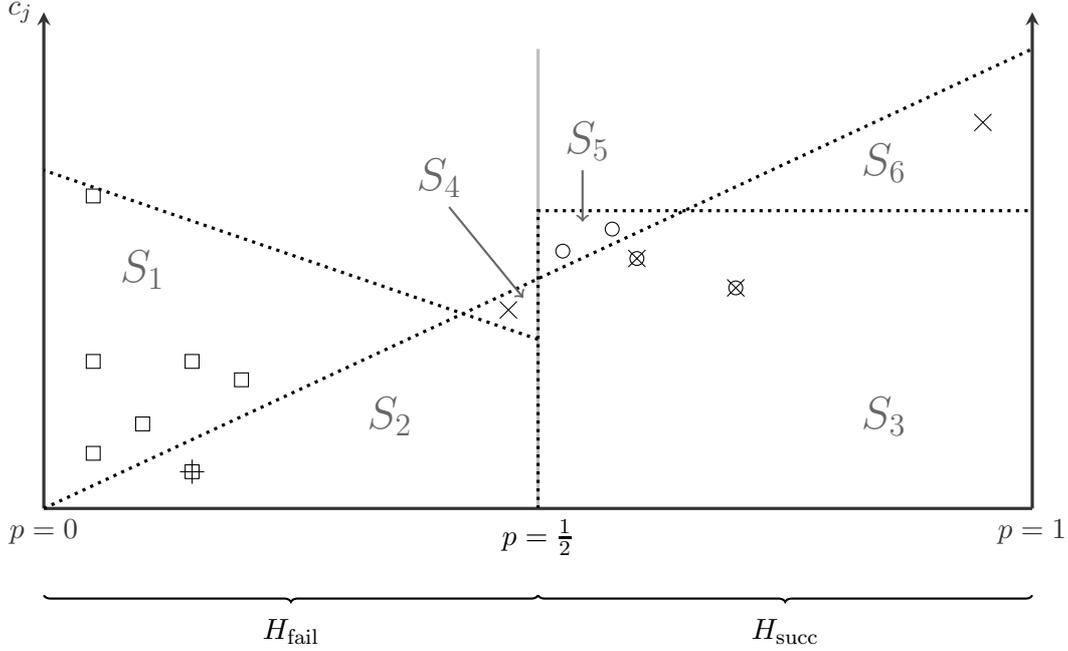
\begin{figure}
    \centering
    \begin{tikzpicture}[scale=0.65]
    \pgfmathsetmacro{\ysc}{0.75}
    \tikzset{bar/.style={very thick, color=black!80}}
    \tikzset{outerbar/.style={bar, very thick}}
    \tikzset{pbar/.style={very thick, gray!50}}
    \tikzset{pqline/.style={dotted, very thick, black}}
    \tikzset{hsline/.style={pqline}}
    \tikzset{setlabelline/.style={color=black, thick}}
    \tikzset{plotbase/.style={mark options={scale=3.5}, only marks}}
    \tikzset{a1/.style={plotbase,mark=x}} % items sel by Alg1 not in s0
    \tikzset{s0/.style={plotbase,mark=square,mark options={scale=2}}} % items sel by alg0 with p < 1/2
    \tikzset{r1/.style={plotbase,mark=o,mark options={scale=2}}} % smallest items with p>1/2
    \tikzset{b1/.style={plotbase,mark=+,mark options={scale=3.5}}} % items sel by Alg1 in s0
    \tikzset{a1r1/.style={plotbase,mark options={scale=2}}} % items both in a1 and r1
    
    % draw coordinate system
    \draw[outerbar, -stealth] (0,0) node[below]{$p = 0$} -- (0, 13.5*\ysc) node[left]{$c_j$};
    \draw[outerbar] (0,0) -- (20, 0);
    \draw[outerbar, -stealth] (20, 0) node[below]{$p = 1$} -- (20, 13.5*\ysc);
    \draw[pbar] (10,0) node[below,color=black]{$p =  \frac12$} -- (10,12.5*\ysc);
    \draw[pqline] (0,0) -- (20, 12.5*\ysc);
    %\draw[pqline] (0,9) -- (20,0);
    %\draw[pqline] (0,8) -- (20,8);
    \pgfmathsetmacro{\qplinecost}{9.2};
    \draw[pqline] (10,0.5*\qplinecost*\ysc) -- (0, \qplinecost*\ysc);
    \pgfmathsetmacro{\hspmax}{8.1};
    \draw[hsline] (10,0) -- (10, \hspmax *\ysc) -- (20, \hspmax * \ysc);
    %\draw[hsline] (10, \hspmax * \ysc) -- (20, \hspmax * \ysc);
    
    % draw points
    \coordinate (p1) at (14,6*\ysc);
    \draw plot[a1] coordinates{(p1)};
    \draw plot[r1] coordinates{(p1)};
    
    \coordinate(p2) at (19,10.5*\ysc);
    \draw plot[a1] coordinates{(p2)};
    
    \coordinate(p3) at (10.5,7*\ysc);
    \coordinate(p4) at (11.5,7.6*\ysc);
    \coordinate(p5) at (12,6.8*\ysc);
    \draw plot[r1] coordinates{(p3)};
    \draw plot[r1] coordinates{(p4)};
    \draw plot[r1] coordinates{(p5)};
    \draw plot[a1] coordinates{(p5)};
    
    \coordinate(p6) at (9.4,5.4*\ysc);
    \draw plot[a1] coordinates{(p6)};
    
    \coordinate(p7) at (3,1*\ysc);
    \draw plot[b1] coordinates{(p7)};
    \draw plot[s0] coordinates{(p7)};
    
    \coordinate(p8) at (1,1.5*\ysc);
    \coordinate(p9) at (2,2.3*\ysc);
    \coordinate(p10) at (3,4*\ysc);
    \coordinate(p12) at (1,4*\ysc);
    \coordinate(p13) at (4,3.5*\ysc);
    \coordinate(p11) at (1,8.5*\ysc);
    \foreach \xy in {p8,p9,p10,p11,p12,p13}
    {
        \draw plot[s0] coordinates{(\xy)};
    }

    % draw set labels
    \draw (2,6.5*\ysc) node {\color{black!60}\LARGE $\VizSet_1$};
    \draw (7,2.5*\ysc) node {\color{black!60}\LARGE $\VizSet_2$};
    \draw (17,2.5*\ysc) node {\color{black!60}\LARGE $\VizSet_3$};
    %\draw (12,2*\ysc) node {\huge $\VizSet_3$};
    \draw (8,9*\ysc) node {\color{black!60}\LARGE $\VizSet_4$};
    \draw [setlabelline,->,draw=black!60] (8.2, 8.2*\ysc) -- (9.7, 4.3);
    \draw (11,10*\ysc) node [] {\color{black!60}\LARGE $\VizSet_5$};
    \draw [setlabelline,->,draw=black!60] (10.9, 9.2*\ysc) -- (10.9, 7.8*\ysc);
    \draw (17,9.4*\ysc) node {\color{black!60}\LARGE $\VizSet_6$};

    % draw half space braces
    \begin{scope}[yshift = -1.5cm]
        \draw[thick,decoration={brace,mirror,raise=5pt},decorate] (0,0) -- node[draw=none,below=10pt] {$\Hsp_{\Fail}$} (10,0);
        \draw[thick,decoration={brace,mirror,raise=5pt},decorate] (10,0) -- node[draw=none,below=10pt] {$\Hsp_{\Succ}$} (20,0);
    \end{scope}
    
    \end{tikzpicture}
    \caption{Visualization of some of the sets used in the proof of \Cref{lemm:Alg2:Phase1}.
    The slanted lines mark a constant $c_j$-$p_j$- or $c_j$-$(1-p_j)$-ratio, the decision criteria for \ALGSo and \ALGSz.\newline
    Markers symbolize tests in $c_j$-$p_j$-space:
    \copy\aMark{} stands for tests in $\UnqSet_\HSsucc$, \copy\rMark{} for $\ChSet_\HSsucc$ and \copy\sMark{} for $\PrivSet_\HSfail$. The symbol \copy\bMark{} stands for a test performed by \ALGSo{} not in $\UnqSet_\HSsucc$ and thus contained in $\PrivSet_\HSfail$. Combinations of markers stand for tests contained in multiple sets.
    The set $\ChSet_{\Succ}$ corresponds to $\VizSet_3 \cup \VizSet_5$, $\PrivSet_{\Fail}$ to $\VizSet_1 \cup \VizSet_2$, and $\UnqSet_{\Succ}$ to $\VizSet_3 \cup \VizSet_4 \cup \VizSet_6$. Here, $n_i = 11$ with $|\ChSet_\HSfail| =  7$ and $|\ChSet_\HSsucc| = 4$.}
    \label{fig:Alg2:set-defs}
\end{figure}
For $h \in \CB{\Fail, \Succ}$, let $\PrivSet_h$ be the subset of $|\ChSet_h|$ tests from $\Hsp_h$ with the smallest values of $\Sig{h}^{-1}$. Note that $\PrivSet_\Succ \cap \PrivSet_{\Fail} = \emptyset$.
One can think of $\PrivSet_h$ as a set \enquote{private} tests for each of the two algorithms.
In \Cref{fig:Alg2:set-defs}, $\PrivSet_\Fail$ is visualized as $\VizSet_1 \cup \VizSet_2$.
Let $\bar{h}$ be the element of $\CB{\Fail, \Succ}$ that is not $h$.
Consider now $\UnqSet_h$, defined as the set of $|\ChSet_h|$ tests from $N \setminus \PrivSet_{\bar{h}}$ with smallest values of $\Sig{h}^{-1}$ (this is well-defined since $\ChSet_h \subseteq N \setminus \PrivSet_{\bar{h}}$).
In~\Cref{fig:Alg2:set-defs} (which focuses on $h=\Succ$), $\UnqSet_\Succ$ corresponds to $\VizSet_3 \cup \VizSet_4 \cup \VizSet_6$.
This definition implies that $\ALG_h$ picks tests exclusively from $\UnqSet_h \cup \PrivSet_{\bar{h}}$ until it has queried all $|\ChSet_h|$ tests from $\UnqSet_h$: If a test is selected that is neither in $\UnqSet_h$ nor in $\PrivSet_{\bar{h}}$, all tests in $\UnqSet_h$ must have been already chosen since they have a smaller value of $\Sig{h}^{-1}$.
%The construction of $\UnqSet_\Succ$ and $\UnqSet_\Fail$ can also be thought of in the following way: For \ALGSo we include tests it selects from \enquote{its} half-space $\Hsp_\HSsucc$ ($\VizSet_3 \cup \VizSet_6$).
%If \ALGSo selects a test from $\Hsp_\HSfail$, we include it unless it is contained in $\PrivSet_\Fail$ ($\VizSet_4$).  
%We include the first $|\ChSet_\HSsucc|$ tests chosen by \ALGSo that match these criteria.
%These sets are visualized in a graphical example in \Cref{fig:Alg2:set-defs}.

Let's attempt a proof similar to the idea above. Consider the first time when all tests from both $\UnqSet_\Succ$ and $\UnqSet_\Fail$ have been inspected, and let $u \in \CB{\Succ, \Fail}$ such that the last test chosen at that time was selected by $\ALG_u$.
If, as we will prove in the following lemma, the sets $\UnqSet_\Succ$ and $\UnqSet_\Fail$ are disjoint, then we have inspected at least $|\UnqSet_\Succ| + |\UnqSet_\Fail| = n_i$ unique tests and thus phase~1 must be finished.
Fortunately, we can still handle the cost: Observe that \Cref{lemm:Alg2:helper} can again be applied to bound the cost of tests in $\UnqSet_u$.
We can also bound the cost of tests that $\ALG_u$ inspects from $\PrivSet_{\bar{u}}$ using \Cref{lemm:Alg2:helper} applied to $\Hsp_{\bar{u}}$ where $\bar{u}$ is the element of $\CB{\Fail,\Succ}$ that is not $u$.
The total cost incurred by the other sub-algorithm can then be upper-bounded by the same amount thanks to the round-robin property. This approach is formalized in the following lemma.

\begin{lemma}
\label{lemm:Alg2:Phase1}
    For any $h \in \CB{\Fail,\Succ}$, it holds that
    $$\sum_{\tau=1}^{\tau_1} \cost^\TwoRR_\tau(\ALG_h,I)\leq 2\cdot\Cost{}{\OPT}{\Inst}.$$
\end{lemma}

\begin{proof}
    First note that, as in \Cref{eq:phase1-opt}, it holds that
    \begin{equation}\label{eq:RR2Phase1Opt}
        \Cost{}{\OPT}{\Inst} \geq \sum_{j \in \ChSet} c_j\, .
    \end{equation}
    We first observe that $\RB{\UnqSet_h \cap \Hsp_h} \subseteq \PrivSet_h $. Indeed, $\UnqSet_h$ contains tests with the smallest values of $\Sig{h}^{-1}$, just like $\PrivSet_h$. In the case of $\PrivSet_h$, these tests are only chosen from $\Hsp_h$, but in the case of $\UnqSet_h$ from $N\setminus \PrivSet_{\bar{h}}\supseteq \Hsp_h$. Symmetrically, $\RB{\UnqSet_{\bar{h}} \cap \Hsp_{\bar{h}}} \subseteq \PrivSet_{\bar{h}} $. We further observe that $\UnqSet_{\bar{h}}$ and $\UnqSet_h$ are disjoint sets:
    Consider test $j \in \UnqSet_{\bar{h}}$. If $j \in \Hsp_h$, then, by definition of $\UnqSet_{\bar{h}}$, it follows that $j \notin \PrivSet_h$. Thus, $j \notin \RB{\UnqSet_h \cap \Hsp_h}$, implying $j \notin \UnqSet_h$.
    In the case that $j \in \Hsp_{\bar h}$, we have that $j \in \PrivSet_{\bar h}$ and thus, by definition of $\UnqSet_h$, again $j \notin \UnqSet_h$.
    
    Recall the following important observation: By definition of \Cref{alg:RR} and the sets above, $\ALG_h$ exclusively selects tests from $\UnqSet_h \cup \PrivSet_{\bar{h}}$ until all tests in $\UnqSet_h$ have been inspected. The same holds, of course, for $\ALG_{\bar{h}}$ and $\UnqSet_{\bar{h}}$.
    In particular, this means that, until the beginning of step $\tau_1$, there exists an index $u \in \CB{\Fail,\Succ}$ and a test from $\UnqSet_u$ that has not yet been selected by \TwoRR{}. 
    This is true since at most $n_i - 1 = |\UnqSet_\Fail| + |\UnqSet_\Succ| - 1$ tests have been conducted in total so far and $\UnqSet_\HSfail \cap \UnqSet_\HSsucc = \emptyset$.
    %Let $\bar{u}$ be the element of $\CB{\Fail,\Succ}$ that is not $u$.
    
    Consider again $\ALG_h$.
    If $\ALG_h$ was not chosen by \TwoRR to conduct any test in phase~1, the inequality in the lemma statement is trivially true.
    Otherwise, let $\hat{\tau}_1 \leq \tau_1$ be the last step in which \TwoRR chose $\ALG_h$ to conduct a test before the end of phase 1.
    As in the proof of \Cref{lemm:RR:Phase1}, let $j^\star$ be the minimum value such that $\SigVal{u}{j^\star}$ has not been conducted up to and including step $\hat{\tau}_1 -1$ (i.e., $j^\star$ is the next test that $\ALG_u$ would conduct). Note that, by choice of $u$ and the arguments above, it must hold that $j^\star \in \UnqSet_u \cup \PrivSet_{\bar{u}}$.
    
    So, using the order in which $\ALG_u$ conducts tests, we have
    \begin{equation}\label{eq:RR2Phase1aux}
        \sum_{\tau=1}^{\hat{\tau}_1-1} \cost^\TwoRR_\tau(\ALG_u,I)\leq \sum_{j \in \UnqSet_u} c_{j} + \sum_{j \in \PrivSet_{\bar{u}}} c_{j} -c_{j^\star} \,.
    \end{equation}
    The way the test for step $\hat{\tau}_1$ is selected by \TwoRR, 
    \begin{align}
        &\sum_{\tau=1}^{\tau_1} \cost^\TwoRR_\tau(\ALG_h,I)=\sum_{\tau=1}^{\hat{\tau}_1} \cost^\TwoRR_\tau(\ALG_h,I) \notag\\\leq &\sum_{\tau=1}^{\hat{\tau}_1-1} \cost^\TwoRR_\tau(\ALG_u,I) + c_{j^\star} \leq \sum_{j \in \UnqSet_u} c_{j} + \sum_{j \in \PrivSet_{\bar{u}}} c_{j},\label{eq:RR2Phase1main} \,
    \end{align}
    where the third step follows from \Cref{eq:RR2Phase1aux}.
    
    Lastly, applying \Cref{lemm:Alg2:helper} with $F = \ChSet_u$ and $G = N \setminus \PrivSet_{\bar{u}}$ (and thus $G_{|\ChSet_u|} = \UnqSet_u$) as well as with $F = \ChSet_{\bar{u}}$ and $G = H_{\bar{u}}$ (and thus $G_{|\ChSet_{\bar{u}}|} = \PrivSet_{\bar{u}}$) yields
    \begin{equation}\label{eq:RR2Phase1fac2}
        \sum_{j \in \UnqSet_u} c_j \leq 2 \cdot \sum_{j \in \ChSet_u} c_j \text{  and  } \sum_{j \in \PrivSet_{\bar{u}}} c_j \leq 2 \cdot \sum_{j \in \ChSet_{\bar{u}}} c_j \, .
    \end{equation}
    Combining \Cref{eq:RR2Phase1main}, \Cref{eq:RR2Phase1fac2}, the fact that $\ChSet = \ChSet_\Succ \, \dot\cup \, \ChSet_\Fail$, and \Cref{eq:RR2Phase1Opt} implies the lemma.
\end{proof}
 
We now have all the ingredients required for proving \Cref{thm:TwoRR}.

\begin{proof}[Proof of \Cref{thm:TwoRR}.]
    The proof works in analogy to that of \Cref{thm:TwRR}. Consider instance $I'$ and $i\in\{1,\dots,B-1\}$ such that $I^\TwoRR_{\tau_1}=I'\wedge f(x)=i$ with nonzero probability. 
    Further, let $h \in \CB{\Fail, \Succ}$. Again denoting by $\PhaseTwoActive{h}$ the event that $\ALG_h$ conducts a test in phase 2, we distinguish whether or not $\PhaseTwoActive{h}$ occurs with nonzero probability under the condition $I^\TwoRR_{\tau_1}=I'\wedge f(x)=i$. In the first case, we use Lemmas~\ref{lemm:Alg2:Phase1} and~\ref{lemm:RR:Phase2} (in comparison with Theorem~\ref{thm:TwRR}, the former instead of Lemma~\ref{lemm:RR:Phase1}); in the second case, we only need Lemma~\ref{lemm:Alg2:Phase1} (again instead of Lemma~\ref{lemm:RR:Phase1}). We then combine the inequalities for the two cases and combine the resulting inequalities across all $I'$, $i$, and $h$.
    
    We again start with the case $\Ptyc{\PhaseTwoActive{h}}{I^\TwoRR_{\tau_1}=I'\wedge f(x)=i} > 0$. By first applying \Cref{lemm:RR:Phase2} with \TwoRR taking the role of \GenRR (meaning that $s$, $\tau_1$, $\tau_2$, and $E_h$ take the role of $s'$, $\tau_1'$, $\tau_2'$, and $E_h'$, respectively) and then applying Lemma~\ref{lemm:Alg2:Phase1}, we obtain
    \begin{align}
        \Ptyc{\PhaseTwoActive{h}}{I^\TwoRR_{\tau_1}=I'\wedge f(x)=i}\;\cdot\;&\Exc{\sum_{\tau=1}^{\tau_2}\cost^{\TwoRR}_\tau(\ALG_{h},I)}{I^\TwoRR_{\tau_1}=I'\wedge f(x)=i\wedge\PhaseTwoActive{h}}\notag\\
        \leq  \Ptyc{\PhaseTwoActive{h}}{I^\TwoRR_{\tau_1}=I'\wedge f(x)=i}\;\cdot\;&\Exc{\sum_{\tau=1}^{\tau_1}\cost^{\TwoRR}_\tau(\ALG_{s},I)}{I^\TwoRR_{\tau_1}=I'\wedge f(x)=i\wedge\PhaseTwoActive{h}} \notag\\
        + \;& \Exc{\Cost{}{\OPT}{\Inst}}{I^\TwoRR_{\tau_1}=I'\wedge f(x)=i}\notag\\
        \leq \Ptyc{\PhaseTwoActive{h}}{I^\TwoRR_{\tau_1}=I'\wedge f(x)=i} \;\cdot\;&2\cdot \Exc{\Cost{}{\OPT}{\Inst}}{I^\TwoRR_{\tau_1}=I'\wedge f(x)=i\wedge\PhaseTwoActive{h}}\notag\\
        +\;& \Exc{\Cost{}{\OPT}{\Inst}}{I^\TwoRR_{\tau_1}=I'\wedge f(x)=i}.\label{eq:rr2-thm-phasetwoactive}
    \end{align}
    For the case that $\Ptyc{\PhaseTwoInactive{h}}{{I^\TwoRR_{\tau_1}=I'\wedge f(x)=i}} > 0$, we only apply Lemma~\ref{lemm:Alg2:Phase1} and obtain
    \begin{align}
    &\Ptyc{\PhaseTwoInactive{h}}{I^\TwoRR_{\tau_1}=I'\wedge f(x)=i}\cdot\Exc{\sum_{\tau=1}^{\tau_2}\cost^{\TwoRR}_\tau(\ALG_{h},I)}{I^\TwoRR_{\tau_1}=I'\wedge f(x)=i\wedge\PhaseTwoInactive{h}}\notag\\
    = \;&\Ptyc{\PhaseTwoInactive{h}}{{I^\TwoRR_{\tau_1}=I'\wedge f(x)=i}}\cdot\Exc{\sum_{\tau=1}^{\tau_1}\cost^{\TwoRR}_\tau(\ALG_{h},I)}{I^\TwoRR_{\tau_1}=I'\wedge f(x)=i\wedge\PhaseTwoInactive{h}}\notag\\
    \leq \;&\Ptyc{\PhaseTwoInactive{h}}{{I^\TwoRR_{\tau_1}=I'\wedge f(x)=i}}\cdot2\cdot\Exc{\Cost{}{\OPT}{\Inst}}{I^\TwoRR_{\tau_1}=I'\wedge f(x)=i\wedge\PhaseTwoInactive{h}}.\label{eq:rr2-thm-phasetwoinactive}
    \end{align}
    
For the next step, again define $\mathcal{E}$ to be the set of events from $\{\PhaseTwoActive{h},\PhaseTwoInactive{h}\}$ with nonzero probability conditional on $I^\TwoRR_{\tau_1}=I'\wedge f(x)=i$. Then, by the law of total expectation, we obtain
    \begin{align*}
    &\Exc{\sum_{\tau=1}^{\tau_2}\cost^{\TwoRR}_\tau(\ALG_{h},I)}{I^\TwoRR_{\tau_1}=I'\wedge f(x)=i}\\
        = \;& \sum_{E\in\mathcal{E}}\Ptyc{E}{I^\TwoRR_{\tau_1}=I'\wedge f(x)=i}\cdot\Exc{\sum_{\tau=1}^{\tau_2}\cost^{\TwoRR}_\tau(\ALG_{h},I)}{I^\TwoRR_{\tau_1}=I'\wedge f(x)=i\wedge E}\\
        \leq \;& 2\cdot\Exc{\Cost{}{\OPT}{\Inst}}{I^\TwoRR_{\tau_1}=I'\wedge f(x)=i}+\Exc{\Cost{}{\OPT}{\Inst}}{I^\TwoRR_{\tau_1}=I'\wedge f(x)=i}\\
        = \;& 3\cdot\Exc{\Cost{}{\OPT}{\Inst}}{I^\TwoRR_{\tau_1}=I'\wedge f(x)=i}.
    \end{align*}
    In the second step, we use Equations~\eqref{eq:rr3-thm-phasetwoactive} (if $\PhaseTwoActive{h}\in\mathcal{E}$) and~\eqref{eq:rr3-thm-phasetwoinactive} (if $\PhaseTwoInactive{h}\in\mathcal{E}$) as well as the law of total expectation again.    
    
    By applying the law of total expectation over all $I'$ and $i\in\{1,\dots,B-1\}$ such that $I^\TwoRR_{\tau_1}=I'$ and $f(x)=i$ with nonzero probability, we obtain from the previous inequality that
    \begin{equation}
        \Ex{\sum_{\tau=1}^{\tau_2}\cost^\TwoRR_\tau\RB{\ALG_h,I}} \leq 3 \cdot \Ex{\Cost{}{\OPT}{\Inst}} \, .\label{eq:Alg2-thm-alphas}
    \end{equation}
    Now using that
    \begin{equation*}
        \cost(\TwoRR,I)=\sum_{\tau=1}^{\tau_2}\cost^\TwoRR_\tau(\ALGSz,I)+\sum_{\tau=1}^{\tau_2}\cost^\TwoRR_\tau(\ALGSo,I),    
    \end{equation*}
    we get from Inequality~\eqref{eq:Alg2-thm-alphas} that
    \begin{align*}
        \Ex{\cost(\TwoRR,I)}&\leq 6\cdot\Ex{\Cost{}{\OPT}{\Inst}},
    \end{align*}
    which completes the proof.
\end{proof}

\section{Adaptivity Gap}
\label{sec:CFA:adap}
We show the following lower bound on the adaptivity gap of SSC, establishing a tight bound in the unit-cost $k$-of-$n$ case together with the upper bound due to Grammel et al.~\cite{GrammelHKL:22}. In fact, the type of instance we use has already been used in that work, albeit for a different purpose.

\begin{theorem}\label{thm:agaptivity}
    For any $\varepsilon > 0$, there exists a unit-cost \KofN{} instance \Inst{} such that
    \begin{equation}\label{eq:adapt}
        \mathbb{E}[\cost(\overline{\OPT},{\Inst})] \geq \RB{\frac{3}{2} - \varepsilon} \cdot \mathbb{E}[\cost(\OPT,{\Inst})] \, ,
    \end{equation}
    where $\overline{\OPT}$ and $\OPT$ are optimal non-adaptive and adaptive strategies for instance $\Inst$, respectively.
\end{theorem}

\begin{proof}
To simplify the exposition, we first consider the slightly more general setting in which each success probability may take any value in $[0,1]$, rather than $(0,1)$. In this setting, we still consider the value of a test $j$ unknown unless conducted, even if $p_j\in\{0,1\}$. After showing the statement in this generalized setting, we use a continuity argument to lift the statement to the original setting.

Consider the following unit-cost \KofN{} instance $I_m$ parameterized by $m \in \mathbb{N}$.
Here, $n=2m+1$ and $k = m+1$, and there are $m$ tests $j\in\CB{1, \ldots, m}$ with $p_j = 1$, $m$ tests $j\in\CB{m+1, \ldots, 2m}$ with $p_j=0$, and a single additional test $n$ with $p_n = \frac{1}{2}$.
Observe that the function value of $f$ is $1$ (i.e., the number of successes lies in the first interval) if and only if $x_n = 0$. The instance is visualized in \Cref{fig:gap-instance}.

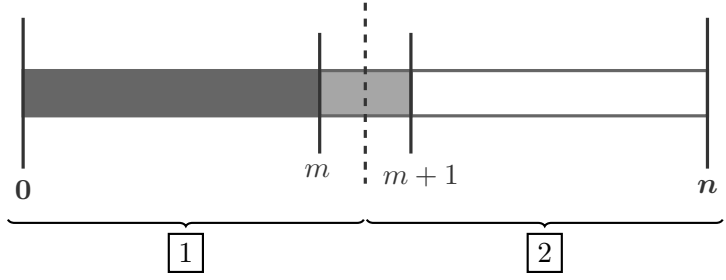
\begin{figure}
    \centering
    \begin{tikzpicture}[]
    \tikzset{bar/.style={very thick, color=black!80}}
    \tikzset{outerbar/.style={bar, very thick, node font=\boldmath}}
    \tikzset{conint/.style={pattern=north east lines}}
    
    \filldraw[color=black!60, very thick] (0, .3) rectangle (3.9, -.3);
    \filldraw[color=black!60, very thick, fill=white] (5.1, .3) rectangle (9, -.3);
    \filldraw[color=black!60, very thick, fill=black!35] (3.9, .3) rectangle (5.1, -.3);
    
    \draw[outerbar] (0, 1) -- (0, -1) node[below]{$0$};
    \draw[outerbar] (9, 1) -- (9, -1) node[below]{$n$};
    \draw[bar] (3.9, .8) -- (3.9, -.8) node[below]{$m\,$};
    \draw[bar] (5.1, .8) -- (5.1, -.8) node[below]{$\,\,\,\,m+1$};
    \draw[bar,dashed] (4.5, 1.2) -- (4.5, -1.2);
    \begin{scope}[yshift = -1.5cm]
        \draw[thick,decoration={brace,mirror,raise=5pt},decorate] (-.2,0) -- node[rectangle,draw=black,below=10pt] {1} (4.48,0);
        \draw[thick,decoration={brace,mirror,raise=5pt},decorate] (4.52,0) -- node[rectangle,draw=black,below=10pt] {2} (9.2,0);
    \end{scope}

    \end{tikzpicture}

    \caption{Visualization of the adaptivity gap instance.
    The black bar shows $m$ tests with success probability $1$. The white bar shows $m$ tests with success probability $0$.
    The gray bar corresponds to test $n$ with $p_n = \nicefrac12$.
    Note that the interval of the overall score depends only on the realization of the gray test.
    Any algorithm must inspect the gray test and at least $m$ other tests.}
    \label{fig:gap-instance}
\end{figure}

We obtain an (in fact, optimal) \emph{adaptive} strategy by first inspecting test $n$.
If $n$ is a success, we perform tests $1, \ldots, m$ to verify $f\RB{x} = 2$.
Otherwise, we inspect tests $m+1, \ldots, 2m$ and confirm that $f\RB{x} = 1$.
For any realization $x$, this incurs a total cost of $m+1$.

Consider any \emph{non-adaptive} strategy, which can be viewed as a permutation of the $n$ tests.
If the strategy chooses a test from $\CB{1, \ldots, m}$ last, i.e., as $n$-th test in the permutation, then conditioned on $x_n = 1$ (i.e., with probability $\nicefrac12$), it terminates only after having performed that last test, incurring a total cost of $2m+1$.
Conversely, a strategy that chooses a test from $\CB{m+1, \ldots, 2m}$ last pays $2m+1$ whenever $x_n = 0$, which happens again with probability $\nicefrac12$.
Lastly, a strategy that inspects test $n$ last always has to perform all tests.
This implies that the expected cost of any non-adaptive strategy is at least $\nicefrac12 \cdot\RB{m+1} + \nicefrac12 \cdot\RB{2m+1} = \nicefrac32 \cdot m + 1$.

Combining the upper bound on the (expected) cost of the adaptive strategy and the lower bound on the expected cost of any non-adaptive strategy, we obtain
\begin{equation}
    \frac{\mathbb{E}\big[\cost(\overline{\OPT},{I_m})\big]}{\mathbb{E}\big[\cost(\OPT,{I_m})\big]}\geq \frac{\frac32\cdot m + 1}{m+1} \; \underset{m\to\infty}{\longrightarrow} \; \frac32\,\label{eq:adaptivity-gap}
\end{equation}
as a lower bound on the adaptivity gap in the generalized setting with $p_j\in[0,1]$ for all tests $j$.

Towards lifting the bound to the original setting, first observe that, for any instance $I$ and any algorithm \ALG,
\begin{align*}
    \Ex{\Cost{}{\ALG}{\Inst}} \;&= \sum_{y\in\{0,1\}^N} \Pty{x=y}\cdot\Exc{\Cost{}{\ALG}{\Inst}}{x = y}\\
    &=\sum_{y\in\{0,1\}^N} \prod_{j\in N: y_j=0}(1-p_j)\prod_{j\in N: y_j=1}p_j\cdot\Exc{\Cost{}{\ALG}{\Inst}}{x = y}.
\end{align*}
Note that this expression is continuous in the success probabilities $p_j$.
Since $\OPT$ minimizes the expected cost over finitely many algorithms, $\mathbb{E}[\cost(\OPT,{\Inst})]$ is also continuous with respect to these perturbations---even though the optimal strategy may change as the parameters of the instance change.
The analogous statement for $\mathbb{E}[\cost(\overline{\OPT},{\Inst})]$ follows from the same argument.
This means that, for any $\varepsilon'>0$, there exists an instance $\Inst'_m$ in the original setting such that the adaptivity gap on that instance is by at most $\varepsilon'$ larger than on $I_m$. This, together with \Cref{eq:adaptivity-gap}, implies the claim.
\end{proof}

\section{Conclusion}
\label{sec:conclusion}

An obvious direction for future research is further improving the approximation guarantees shown in this paper. It is possible that the analyses for \TwRR and \TwoRR provided here are not tight. Note that, however, the instance given in the proof of Theorem~\ref{thm:agaptivity} provides a lower bound of $2$ on the approximation guarantees of both these algorithms. Another direction is finding simple algorithms with small approximation guarantees for the extensions considered by Ghuge et al.~\cite{GhugeGN:21}, e.g., the weighted version.

In fact, it is not known whether the SSC problem is NP-hard. The same question for similar problems, e.g., Stochastic Function Evaluation of Boolean read-once formulas, has now been open for several decades.

A more structural than algorithmic direction is settling the precise adaptivity gap beyond the unit-cost $k$-of-$n$ case.

\section*{Acknowledgments}
The authors wish to thank Marcus Kaiser for helpful discussions and an anonymous reviewer for many comments that improved the quality of our presentation.

\bibliographystyle{plain}
\bibliography{ssc}

\end{document}